\documentclass[12pt]{article}
\usepackage[margin=1in]{geometry}
\usepackage{amsmath,amssymb,amsthm,mathtools}
\usepackage{booktabs,makecell,array,longtable}
\usepackage{graphicx}
\usepackage[table]{xcolor}
\usepackage{natbib}
\usepackage[hidelinks,hypertexnames=false]{hyperref}
\usepackage{enumitem}
\usepackage{pdflscape}
\usepackage{float}

\newtheorem{assumption}{Assumption}
\newtheorem{theorem}{Theorem}
\newtheorem{proposition}{Proposition}

\newcommand{\E}{\mathbb E}
\newcommand{\Pbb}{\mathbb P}

\newcommand{\G}{\mathcal G}

\newcommand{\BaR}{\operatorname{BaR}}

\title{Adaptive AI Delegation under Uncertainty:\\A Bayesian Governance Policy for Sequential Decision Authority}
\author{Matthew Dixon}
\date{\today}

\begin{document}
\maketitle

\begin{abstract}
Organizations increasingly use large language models and agentic AI systems to generate probabilistic assessments and candidate actions in high-consequence settings. This creates a managerial problem that is fundamentally different from prediction: how should organizations allocate decision authority to AI-generated recommendations as evidence quality, uncertainty, and organizational objectives evolve over time? Existing AI governance frameworks emphasize transparency, documentation, oversight, and regulatory compliance, but they provide limited quantitative guidance for dynamically allocating decision authority under uncertainty. To address this challenge, we formulate adaptive AI delegation as a Governance-Aware Partially Observable Markov Decision Process (POMDP) in which Bayesian inference estimates the informational state and sequential optimization determines the appropriate degree of delegated AI authority.

The paper also develops a quantitative validation and benchmarking framework for governance policies. Synthetic stress tests, reported LLM-confidence robustness, forecast-accuracy validation, governance-appetite sensitivity, and fragile-AI early-warning experiments evaluate whether the proposed policy exhibits graceful degradation, robustness to confidence-only perturbations, adaptive delegation under improving evidence quality, and interpretable calibration of institutional conservatism. Finally, the Governance-Aware POMDP is benchmarked against five representative governance strategies operating under identical Bayesian beliefs, information, and governance objectives. The results show that while specialized heuristics remain effective under stationary operating conditions, sequential Bayesian governance provides the strongest general-purpose governance policy across heterogeneous AI-quality regimes by dynamically allocating organizational decision authority under uncertainty.

\end{abstract}


\section{Introduction}
\label{sec:introduction}

\textbf{Motivation.} Artificial intelligence (AI) is increasingly embedded within organizational decision processes, supporting decisions in finance, healthcare, cybersecurity, manufacturing, supply-chain management, and other high-consequence domains. In many settings, AI systems no longer function solely as predictive models; instead, they operate as decision-support agents whose recommendations influence human judgment and organizational action. As organizations adopt increasingly capable foundation models and agentic AI systems, a central managerial challenge emerges: \emph{how should decision authority be allocated between AI-generated intelligence and established organizational decision processes when the reliability of AI is uncertain?}

\textbf{Governance.} Existing approaches to AI governance emphasize model validation, monitoring, documentation, transparency, explainability, regulatory compliance, and post-deployment oversight \citep{nist2023,iso42001,oecd2019,eu_ai_act}. Although these mechanisms are essential, they are largely static. Organizations repeatedly decide whether AI-generated recommendations should be accepted, moderated, or overridden as new evidence becomes available and operating conditions evolve. Governance therefore becomes a sequential organizational decision problem rather than a one-time certification exercise. The relevant managerial question is not simply whether an AI model is accurate on average, but how organizations should dynamically regulate the influence of AI-generated intelligence over time.

\textbf{Delegation.} We argue that this challenge is fundamentally one of adaptive organizational delegation. Classical organizational theory distinguishes information acquisition from decision authority and emphasizes that organizations frequently separate these functions to improve decision quality \citep{simon1947,galbraith1974,jensen1995,aghion1997}. Bayesian decision theory similarly distinguishes inference from action by representing uncertainty through posterior beliefs that evolve as new information arrives \citep{blackwell1962discrete,puterman1994mdp,kaelbling1998}. Rather than allowing AI systems to determine organizational actions directly, organizations require a governance mechanism that continuously determines the appropriate degree of authority assigned to AI-generated intelligence. This mechanism should adapt to evolving evidence while preserving the ability to revert toward independently validated organizational knowledge whenever governance risk becomes sufficiently large.

\textbf{Research Questions.} This paper addresses four related questions. First, how should organizations dynamically allocate decision authority between AI-generated intelligence and existing organizational decision processes? Second, how should uncertainty regarding AI reliability be represented and updated as new information becomes available? Third, can AI governance itself be formulated as a sequential optimization problem rather than as a collection of static validation rules? Fourth, does sequential Bayesian governance provide measurable organizational advantages relative to simpler governance mechanisms such as static delegation, confidence thresholds, Bayesian shrinkage, reliability-only delegation, and rule-based governance?

\textbf{Approach.} To answer these questions, we formulate adaptive AI delegation as a Governance-Aware Partially Observable Markov Decision Process (POMDP). Bayesian filtering estimates posterior beliefs regarding the latent operating environment, while a governance policy maps those beliefs into delegated decision authority assigned to AI-generated recommendations. Unlike conventional POMDP applications, the optimization variable is not the operational decision itself but the degree of organizational authority assigned to AI-generated intelligence. Governance therefore becomes an adaptive organizational capability rather than a collection of manually specified confidence thresholds or approval rules.

\textbf{Benchmarking.} A distinguishing feature of the paper is the development of a quantitative benchmarking framework for AI governance policies. Rather than evaluating the proposed Governance-Aware POMDP in isolation, we compare it against five representative governance mechanisms: Static Delegation, Confidence Threshold, Reliability-Only Delegation, Bayesian Shrinkage, and SR11-7 Style Governance. Every benchmark policy operates under the identical Bayesian filtering framework, identical information set, identical simulated environments, and identical governance objective. The policies differ only in how delegated AI authority is determined. This design isolates the incremental value of sequential governance from differences in forecasting models, Bayesian inference, or utility specification, providing a principled basis for comparing competing governance strategies.

\textbf{Contributions.} This paper makes six principal contributions. First, it reframes AI governance as a dynamic organizational delegation problem in which governance itself becomes the object of sequential optimization rather than post hoc model validation. Second, it develops a parsimonious Bayesian framework that separates inference, validation, governance, and execution while representing organizational uncertainty through posterior beliefs. Third, it formulates a Governance-Aware POMDP that continuously regulates delegated AI authority as uncertainty evolves, replacing heuristic confidence thresholds with principled sequential optimization. Fourth, it establishes theoretical properties of the proposed governance framework, including existence of an optimal governance policy, belief-state sufficiency, monotone adaptation to governance risk, robustness to confidence-only perturbations, and adaptive delegation under improving evidence quality. Fifth, it introduces a comparative benchmarking methodology in which multiple governance policies are evaluated under identical Bayesian beliefs, information, and governance objectives. Finally, it validates the resulting governance architecture through synthetic experiments and historical market replay, demonstrating that adaptive Bayesian governance provides the strongest general-purpose governance policy across heterogeneous AI-quality regimes while remaining competitive in investment performance.

\textbf{Findings.} The empirical results reveal an important distinction between specialized and adaptive governance policies. Bayesian Shrinkage performs best when AI quality is persistently poor because immediate contraction toward the validated reference process minimizes unnecessary AI exposure. However, the Governance-Aware POMDP achieves the highest governance utility across fragile, improving, and high-quality AI regimes, demonstrating that sequential governance becomes increasingly valuable as AI quality evolves over time. Rather than seeking to maximize AI utilization, the proposed framework dynamically allocates organizational authority according to Bayesian evidence quality, governance risk, and institutional objectives.

\textbf{Managerial Perspective.} The proposed framework is intended to augment rather than replace human expertise. Governance decisions remain under organizational control, while AI contributes probabilistic information whose influence varies according to Bayesian evidence, statistical reliability, and governance risk. This perspective is particularly relevant in regulated industries, where organizations must justify not only individual decisions but also the governance processes through which those decisions are made. By embedding governance directly within sequential decision making, organizations can adapt continuously to changing evidence without relying on static approval rules or binary override mechanisms.

\textbf{Paper Organization.} The remainder of the paper is organized as follows. Section~\ref{sec:related_literature} reviews the literature on Bayesian decision making, organizational delegation, AI governance, and human--AI collaboration. Section~\ref{sec:design_principles} formulates the adaptive delegation problem, while Sections~\ref{sec:architecture}--\ref{sec:theory} develop the governance architecture, Bayesian inference framework, governance optimization, Belief-at-Risk, and the principal theoretical properties. Section~\ref{sec:experimental_design} presents the experimental methodology, Section~\ref{sec:synthetic_validation} reports the synthetic governance laboratory and comparative benchmark results, Section~\ref{sec:discussion} discusses managerial implications, and Section~\ref{sec:historical_validation} evaluates the framework using historical market replay. The paper concludes by discussing broader implications for adaptive organizational delegation under uncertainty.

\section{Related Literature}
\label{sec:related_literature}

The proposed framework lies at the intersection of Bayesian sequential decision making, organizational delegation, AI governance, and model risk management. Each literature addresses an important part of decision making under uncertainty, yet none directly formulates the dynamic allocation of decision authority to AI-generated recommendations as the central optimization problem. The governance policy developed here builds on these complementary traditions while treating governance itself as the object of optimization.

\subsection{Bayesian Sequential Decision Making}
Bayesian decision theory provides the statistical foundation for learning under uncertainty by combining prior information with new observations through recursive belief updating. Partially observable Markov decision processes extend this framework to sequential decision problems in which the true system state cannot be observed directly but must instead be inferred from noisy information \citep{smallwood1973optimal,puterman1994mdp,kaelbling1998}. Posterior beliefs become sufficient statistics for future decision making, allowing dynamic optimization to proceed over belief states rather than complete observation histories. The present work adopts this Bayesian perspective but departs from the classical formulation in an important respect. Traditional POMDPs optimize operational actions directly, whereas the proposed governance policy optimizes actions that regulate the influence of AI-generated recommendations before execution. Bayesian inference therefore estimates organizational uncertainty, while governance determines how that uncertainty should affect the delegation of decision authority.

\subsection{Organizational Delegation and Decision Rights}
The organizational literature has long recognized that information acquisition and decision authority need not reside within the same organizational unit. Simon's theory of bounded rationality, Galbraith's information-processing view of organizations, Jensen and Meckling's analysis of organizational architecture, and Aghion and Tirole's distinction between formal and real authority all emphasize that effective performance depends critically on how decision rights are allocated \citep{simon1947,galbraith1974,jensen1995,aghion1997}. These theories largely predate modern AI systems and therefore do not explicitly consider organizations in which one decision-making actor is an adaptive probabilistic AI model. The governance policy developed here extends this literature by treating AI as an organizational actor whose delegated authority is continuously adjusted through Bayesian learning and sequential optimization rather than fixed ex ante.

\subsection{AI Governance and Model Risk}
Rapid advances in AI have generated substantial interest in governance, explainability, accountability, fairness, transparency, auditability, and model oversight. Frameworks such as the NIST AI Risk Management Framework, ISO/IEC AI governance standards, the OECD AI Principles, and the EU AI Act provide important guidance for responsible AI deployment \citep{nist2023,iso42001,oecd2019,eu_ai_act}. Financial institutions also rely on model risk management practices, including SR 11-7, to ensure appropriate validation and oversight of quantitative models \citep{sr117,boe2022aipublicprivate,boe2024aisurvey,fsb2024ai}. These frameworks, however, primarily describe governance processes rather than governance optimization. They specify what organizations should monitor and document but generally do not provide quantitative methods for determining how much authority should be delegated to AI-generated recommendations as evidence quality evolves. The present work complements rather than replaces these governance frameworks by providing a mathematical mechanism through which governance policies can be optimized, calibrated, and empirically validated.

\subsection{Earlier Bayesian Governance Frameworks}
Earlier technical work on agentic AI as a partially observable decision system and on Belief-at-Risk developed the probabilistic foundations on which this paper builds \citep{dixon2026pomdp,dixon2026bar}. Those studies treated AI-generated intelligence as evidence entering a Bayesian state filter and introduced governance diagnostics for quantifying uncertainty, instability, and consequence. The present paper differs in emphasis and scope. It treats the managerial object of interest as delegated decision authority, embeds the diagnostic quantities in a governance utility, and links the resulting policy to organizational delegation, decision rights, and quantitative validation. In this sense, the current contribution extends a technical Bayesian governance framework into a decision-theoretic theory of adaptive AI delegation.

\subsection{Human--AI Decision Making and Algorithmic Reliance}
Empirical research on human--AI decision making shows that decision makers may over-rely on algorithms, under-rely on them after observing errors, or use them differently depending on how recommendations are presented \citep{dietvorst2015algorithm,logg2019algorithm,amershi2019guidelines,bansal2019updates,green2019principles}. Organizational research has similarly examined how algorithms reshape control, expertise, and work design \citep{kellogg2020algorithms,shrestha2019organizational,raisch2021ai}. The proposed framework differs from this literature by focusing not on whether humans accept algorithmic advice, but on how organizations can design a governance policy that dynamically allocates authority between AI-generated recommendations and independently validated reference processes. This perspective treats trust as an endogenous policy variable rather than a behavioural response alone.

\subsection{Dynamic Programming and Approximate Optimization}
The governance policy also draws on dynamic programming and approximate dynamic programming \citep{bellman1957dynamic,bertsekas2017dynamic,powell2022reinforcement}. Classical dynamic programming provides a principled approach for solving sequential optimization problems under uncertainty, but exact solutions often become computationally intractable. The governance policy therefore employs an approximate Bellman recursion over a finite, interpretable governance action space. Importantly, this approximation is applied to governance decisions rather than portfolio allocations, which preserves interpretability while maintaining the sequential nature of the decision problem. Organizations can therefore implement delegation policies that remain transparent, auditable, and operationally practical without sacrificing the logic of dynamic optimization.

\subsection{Positioning}
The proposed framework integrates these streams into a single decision-theoretic architecture. Bayesian inference provides statistically coherent belief updates, organizational theory motivates the allocation of decision authority, AI governance motivates independent validation and oversight, and dynamic programming supplies the optimization machinery. The contribution is to treat governance as a sequential decision problem whose objective is the allocation of organizational authority to AI-generated recommendations under uncertainty. This positioning shifts the focus from AI model performance alone to the behaviour of the governance policy that mediates between AI inference and organizational action.

\section{Design Principles and Formal Delegation Problem}
\label{sec:design_principles}

This section translates the organizational problem introduced above into a formal governance model. The central modeling choice is to optimize delegated decision authority rather than the AI system's operational recommendation directly. This distinction is important because organizations adopting AI must decide not only what the AI recommends, but how much authority that recommendation should receive relative to validated reference processes, internal controls, and institutional risk appetite.

The governance policy is built around five design principles. First, inference and governance are separated because generating information and deciding how much authority should be assigned to that information are distinct organizational tasks. Second, the policy regulates delegated authority rather than selecting the operational action directly. Third, organizational trust should depend primarily on externally validated evidence and governance risk rather than on internally generated confidence declarations. Fourth, governance should adapt as information quality changes. Fifth, the resulting policy should remain interpretable enough to support review, audit, and regulatory oversight. These principles establish the qualitative requirements for adaptive AI governance. The remainder of this section translates them into a quantitative decision framework in which delegated authority becomes the object of optimization.

\subsection{Decision Epoch and Information Inputs}

\textbf{Organizational perspective.} At each decision epoch, the organization observes information from multiple sources. These may include structured observations such as prices, returns, volatility, balance-sheet variables, or operational metrics, as well as unstructured observations such as text, news, filings, analyst reports, or management commentary. The AI system processes this information and produces three conceptually distinct outputs: probabilistic evidence, a candidate recommendation, and a reported confidence score. These outputs should not be collapsed into a single trust measure because they play different roles in the governance process.

\begin{equation}
(q_t,u_t^{AI},c_t)=\mathcal L_\theta(O_t).
\label{eq:section3_llm_outputs}
\end{equation}

\textbf{Role in the framework.} In Equation~\eqref{eq:section3_llm_outputs}, the probability vector $q_t$ is treated as evidence about latent conditions, $u_t^{AI}$ is a candidate recommendation, and $c_t$ is a self-reported confidence measure. The governance policy does not equate confidence with trust. Instead, it treats AI outputs as inputs to a broader organizational decision process that also incorporates Bayesian beliefs, validation history, reference recommendations, and governance risk.

This separation of evidence, recommendation, and confidence provides the informational basis for the delegation decision defined next.

\subsection{Delegated Decision Authority}

\textbf{Decision problem.} Organizations rarely face an all-or-nothing choice between accepting an AI recommendation and rejecting it completely. In high-consequence settings, the relevant managerial question is how much weight to place on AI-generated intelligence relative to an independently validated reference process. The executed recommendation therefore makes delegation explicit.

\begin{equation}
 u_t^{Exec}=\alpha(g_t)u_t^{AI}+(1-\alpha(g_t))u_t^{Ref}.
 \label{eq:exec_intro}
\end{equation}

\textbf{Interpretation.} The governance action $g_t\in\mathcal G$ determines the adaptive delegation weight $\alpha(g_t)\in[0,1]$. When $\alpha(g_t)=0$, the organization follows the reference recommendation. When $\alpha(g_t)=1$, the organization fully delegates to the AI recommendation. Intermediate values represent controlled delegation, allowing the organization to preserve the informational value of AI while retaining a validated fallback.

Delegated execution therefore becomes a continuous organizational design choice rather than a binary approval decision. This representation provides the operational interface between AI-generated intelligence and institutional decision making. Having defined how authority is allocated between AI-generated and reference recommendations, the remaining question is what information the governance policy should use when making this allocation.

\subsection{Governance State}

\textbf{Decision context.} Delegation decisions require more information than the AI recommendation alone. Managers need posterior beliefs about latent conditions, the AI recommendation, the reference recommendation, governance risk, reported confidence, reliability diagnostics, and additional validation signals. These variables define the state on which the governance policy operates.

\begin{equation}
 s_t^G=(b_t,u_t^{AI},u_t^{Ref},\BaR_t,c_t,\tau_t,\eta_t).
 \label{eq:section3_governance_state}
\end{equation}

\textbf{Interpretation.} The governance state separates quantities that are often conflated in AI governance discussions. The posterior belief $b_t$ summarizes Bayesian evidence, $\BaR_t$ summarizes governance uncertainty and consequence, $c_t$ records reported confidence, $\tau_t$ captures statistical reliability from validation history, and $\eta_t$ collects additional diagnostics. This separation supports auditability because each component has a distinct organizational interpretation.

The governance state therefore identifies the information available to the organization before authority is allocated. The next subsection specifies the objective that determines how such states are mapped into governance actions.

\subsection{Governance Objective}

\textbf{Optimization objective.} Governance should optimize long-run organizational performance rather than immediate AI accuracy alone. Each governance action affects current execution, future observations, and subsequent opportunities to learn. The formal problem therefore chooses a governance policy that maximizes expected discounted organizational utility.

\begin{equation}
 \max_{\pi_G}\; \E\left[\sum_{t=0}^{T}\gamma^t r_G(s_t^G,g_t)\right]
 \quad\text{subject to}\quad g_t=\pi_G(s_t^G),
 \label{eq:formal_problem}
\end{equation}

\textbf{Managerial interpretation.} Equation~\eqref{eq:formal_problem} is optimized over governance policies rather than operational policies. This distinction is central to the paper. The organization is not optimizing portfolio trades, diagnoses, or operational actions directly; instead, it optimizes the allocation of decision authority to AI-generated recommendations as uncertainty and governance risk evolve. The reward $r_G$ balances expected organizational value, operational risk, Belief-at-Risk, deviation from the reference process, and intervention costs.

The governance formulation introduced here establishes the decision problem. The next section explains how this policy is operationalized through a modular governance architecture that separates inference, validation, governance, and execution.

Table~\ref{tab:notation} summarizes the notation used throughout the paper. The table appears early because the framework deliberately separates evidence, recommendations, confidence, delegation, and execution.

\begin{table}[H]
\centering
\small
\caption{Notation for the adaptive AI delegation framework. The table distinguishes the principal informational, governance, and execution variables used in the model. This distinction is managerially important because AI evidence, AI confidence, organizational trust, and delegated authority are different quantities.}
\label{tab:notation}
\resizebox{\textwidth}{!}{%
\begin{tabular}{ll}
\toprule
Symbol & Description \\
\midrule
$X_t$ & Latent organizational or market state \\
$O_t$ & Structured and unstructured observations available at time $t$ \\
$q_t$ & AI-generated probabilistic evidence over latent states \\
$u_t^{AI}$ & AI-generated recommendation \\
$c_t$ & Reported AI confidence \\
$b_t$ & Bayesian posterior belief over latent states \\
$u_t^{Ref}$ & Independently validated reference recommendation \\
$u_t^{Exec}$ & Executed recommendation after governance \\
$g_t$ & Governance action \\
$\alpha(g_t)$ & Adaptive delegation weight assigned to AI-generated recommendations \\
$\BaR_t$ & Belief-at-Risk, summarizing uncertainty, instability, and consequence \\
$\tau_t$ & Statistical reliability diagnostic estimated from validation history \\
$s_t^G$ & Governance state \\
$r_G$ & One-period governance utility \\
$V(s_t^G)$ & Governance value function \\
$\lambda_{BaR}$ & Governance appetite parameter on Belief-at-Risk \\
\bottomrule
\end{tabular}%
}
\end{table}

The formulation shifts the focus of the paper from AI model evaluation to organizational decision design. Bayesian inference estimates uncertain latent conditions, governance determines the appropriate degree of delegated authority, and execution implements the resulting governed decision. The next section operationalizes this separation through an explicit governance architecture before developing the Bayesian inference and optimization components in detail.

\section{Governance Architecture}
\label{sec:architecture}

Section~\ref{sec:design_principles} formulated adaptive AI delegation as a formal organizational decision problem. The organization observes AI-generated evidence, candidate recommendations, reported confidence, reference recommendations, and governance diagnostics, but it does not allow the AI system to determine execution directly. Instead, the governance policy allocates decision authority by determining how much influence AI-generated intelligence should receive relative to an independently validated reference process. This section operationalizes that formal problem by describing the architecture through which evidence, validation, risk diagnostics, governance actions, and execution are connected.

The architecture is designed around a simple managerial principle: AI may inform organizational decisions, but governance determines the authority assigned to that information. This distinction is important because the same AI system may provide useful probabilistic evidence while also producing recommendations that require moderation, validation, or rejection. The architecture therefore separates the flow of information from the allocation of decision rights.

Figure~\ref{fig:architecture} summarizes the governance architecture. Structured observations, such as prices, returns, volatility, macroeconomic variables, and portfolio characteristics, are combined with unstructured observations, such as news, research reports, filings, and analyst commentary. The AI system processes these observations and produces three distinct outputs: probabilistic evidence over latent states, a candidate recommendation, and a reported confidence score. These outputs are then mediated by Bayesian filtering, validation diagnostics, Belief-at-Risk, and the governance policy before any executed decision is produced.

\begin{figure}[H]
\centering
\IfFileExists{figures/figure1_bayesian_governance_architecture.png}{\includegraphics[width=0.98\textwidth]{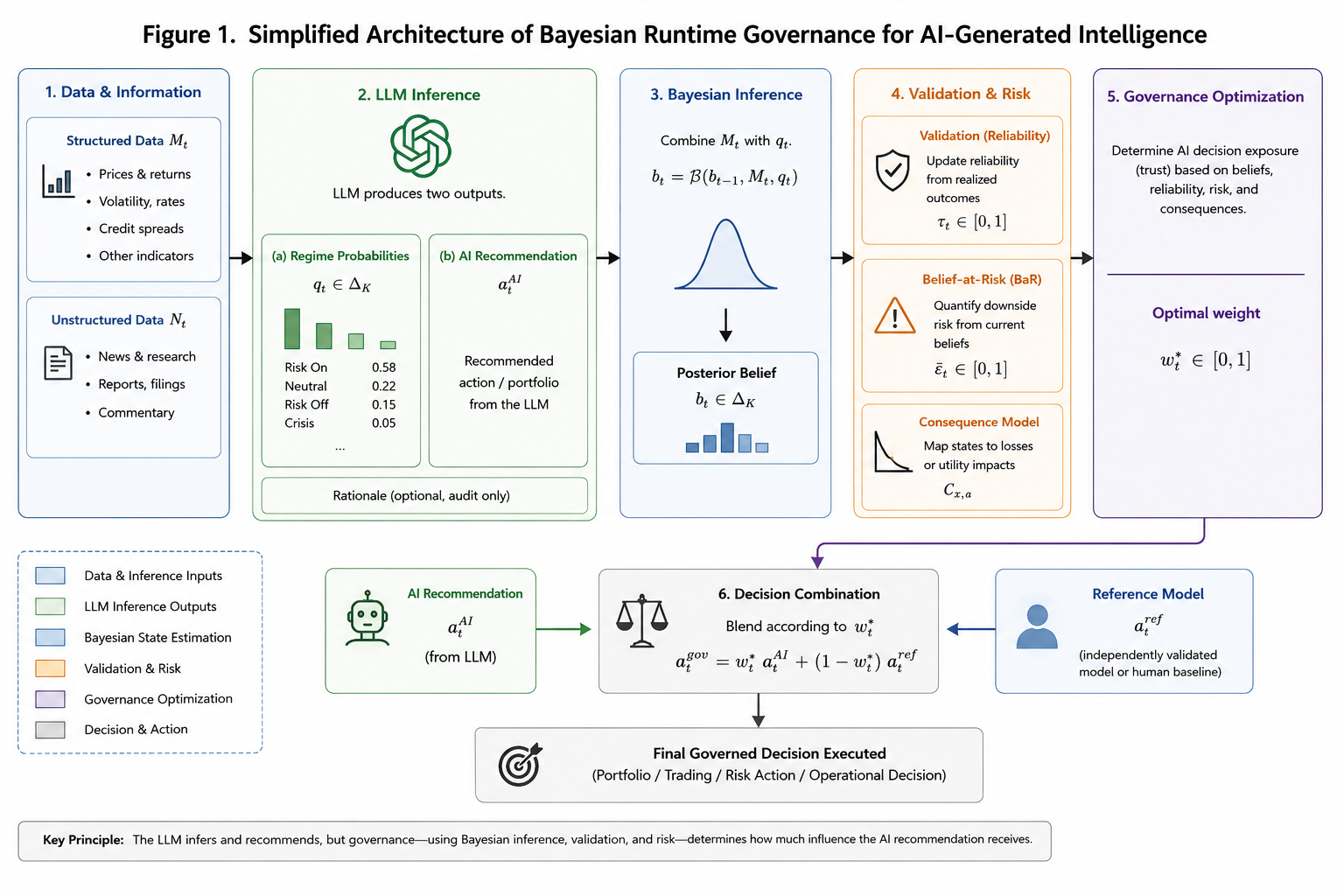}}{\fbox{Missing Figure 1 architecture.}}
\caption{Governance architecture for adaptive AI delegation. Structured and unstructured observations are processed by the AI system to produce probabilistic evidence, candidate recommendations, and reported confidence. Bayesian filtering estimates posterior beliefs; validation and Belief-at-Risk summarize reliability, uncertainty, instability, and consequence; and the governance policy uses sequential optimization to determine the adaptive delegation weight assigned to AI-generated recommendations. Managerially, the architecture separates inference, validation, governance, and execution so that AI recommendations are mediated by a quantitative governance policy rather than directly controlling organizational decisions. This figure provides the conceptual foundation for the belief-state sufficiency and optimal-policy results developed in Sections~\ref{sec:bayesian_pomdp}--\ref{sec:theory}.}
\label{fig:architecture}
\end{figure}

The key feature of the architecture is separation. Bayesian filtering estimates what the organization believes about the latent environment; validation estimates how reliable the AI evidence has been; Belief-at-Risk measures the organizational consequences of acting under uncertain and unstable beliefs; and the governance policy determines delegated decision authority. This separation is consistent with the organizational view that information acquisition and authority allocation are distinct design problems \citep{simon1947,galbraith1974,jensen1995}. It is also consistent with AI governance principles that emphasize oversight and independent validation rather than direct automation of high-consequence decisions \citep{nist2023,iso42001,sr117}.

The architecture should therefore be read as an organizational decision system rather than as a forecasting pipeline. The AI component supplies information and candidate actions, but the organization allocates authority through the governance policy. This distinction is essential for regulated settings, where the ability to audit, explain, calibrate, and override AI-generated recommendations is often as important as the predictive performance of the underlying model.

The architecture also clarifies why reported AI confidence is not sufficient for governance. Reported confidence is only one input to the governance process and is evaluated alongside posterior beliefs, validation-based reliability, Belief-at-Risk, and the reference process. As a result, the organization can distinguish between an AI system that is confident, an AI system that is empirically reliable, and an AI recommendation that should receive delegated authority in the current decision context.

This section establishes the organizational and information architecture, but it does not yet specify how the organization updates its beliefs as evidence arrives. The next section develops that Bayesian inference mechanism. It shows how structured observations, unstructured evidence, and AI-generated probabilistic signals are transformed into posterior belief states. These beliefs then become the informational foundation for the sequential governance policy developed in Section~\ref{sec:optimal_governance}.

\section{Bayesian Governance under Uncertainty}
\label{sec:bayesian_pomdp}

The architecture in Section~\ref{sec:architecture} separates inference from governance. This section develops the inference component. Organizations operating under uncertainty rarely face a shortage of information; instead, they face the more difficult problem of determining how heterogeneous information should influence decisions. Human expertise, established quantitative models, market observations, and AI systems all provide potentially valuable but imperfect signals about future states of the environment. Bayesian decision theory provides a rigorous framework for summarizing such evidence through posterior beliefs \citep{blackwell1962discrete,puterman1994mdp}, while the governance problem determines how those beliefs should influence delegated authority.

The purpose of this section is therefore not to decide how much authority should be delegated to AI. That optimization problem is developed in Section~\ref{sec:optimal_governance}. The purpose here is to define the informational state on which such governance decisions are based.

\subsection{Latent State Representation}

\textbf{Organizational perspective.} The true environment relevant to AI delegation is only partially observable. In the empirical application, this environment is represented by latent market regimes; in other organizational settings, it could represent operational states, risk conditions, demand regimes, or reliability states. A finite latent state representation provides a tractable way to summarize the hidden conditions under which AI-generated recommendations are evaluated.

\begin{equation}
X_t\in\mathcal X=\{1,\ldots,K\}.
\label{eq:latent_state}
\end{equation}

\textbf{Role in the framework.} The state $X_t$ is not the internal state of the AI system. It represents the organization's latent decision environment. Governance acts under uncertainty because the organization observes evidence about $X_t$, but does not observe $X_t$ directly.

This latent-state representation provides the object of inference; the next step is to specify how that object evolves through time.

\subsection{State Dynamics}

\textbf{Sequential evolution.} Organizational environments evolve over time. A Markov transition model provides a parsimonious way to encode persistence and regime changes while retaining the sequential structure required for governance optimization.

\begin{equation}
P_{ij}=\Pbb(X_{t+1}=j\mid X_t=i).
\label{eq:transition_model}
\end{equation}

\textbf{Interpretation.} The transition matrix represents the organization's prior model of how latent states evolve. It supplies historical continuity to the Bayesian filter, preventing governance from reacting only to the most recent AI output.

The transition model supplies the prior component of inference. The current evidence component comes from the AI system and other observations available to the organization.

\subsection{AI Evidence and Recommendations}

\textbf{Information representation.} The AI system converts the available observation set into three distinct objects: probabilistic evidence, a candidate recommendation, and reported confidence. These outputs are separated because they have different governance meanings and should not be collapsed into a single trust score.

\begin{equation}
 (q_t,u_t^{AI},c_t)=\mathcal L_\theta(O_t).
 \label{eq:llm_outputs}
\end{equation}

\textbf{Interpretation.} The probability vector $q_t$ represents AI-generated evidence about latent states, $u_t^{AI}$ represents a proposed action or recommendation, and $c_t$ is the AI system's reported confidence. The governance policy treats reported confidence as an input to be evaluated, not as a sufficient basis for organizational trust.

This decomposition keeps evidence, recommendations, and confidence distinct. Bayesian filtering uses $q_t$ to update beliefs, while the governance policy later evaluates $u_t^{AI}$ and $c_t$ alongside reference recommendations and risk diagnostics.

\subsection{Bayesian Filtering}

\textbf{Predictive belief.} Because governance decisions are sequential, organizational confidence should be updated systematically as evidence arrives. The first step is to propagate the previous posterior belief through the transition model, producing the belief the organization would hold before observing the current AI evidence.

\begin{equation}
\Pi_t=P^\top b_{t-1}.
\label{eq:predictive_belief}
\end{equation}

\textbf{Interpretation.} The vector $\Pi_t$ is the prior predictive belief before incorporating the current AI evidence. It summarizes what the organization would believe based only on past beliefs and the transition model.

The predictive belief preserves continuity with past evidence. The next step incorporates the current AI-generated probabilistic evidence.

\textbf{Bayesian update.} The current evidence $q_t$ is incorporated using Bayes' rule. The update balances the prior predictive belief against the new AI-generated evidence.

\begin{equation}
 b_t(i)=\frac{q_t(i)\Pi_t(i)}{\sum_{j=1}^{K}q_t(j)\Pi_t(j)}.
 \label{eq:bayes_update}
\end{equation}

\textbf{Decision relevance.} The posterior belief $b_t$ summarizes all information relevant for subsequent governance decisions and becomes the sufficient statistic upon which organizational delegation is based. Consistent with the POMDP literature \citep{smallwood1973optimal,kaelbling1998}, the complete observation history is replaced by a recursively updated belief state. Unlike conventional POMDP formulations, however, posterior beliefs do not determine operational actions directly; they provide the informational foundation upon which the governance policy allocates decision authority.

Posterior beliefs alone remain insufficient because organizations must also evaluate the consequences of acting upon uncertain information. Governance therefore requires a diagnostic that distinguishes uncertainty that is merely informational from uncertainty that is unstable and potentially consequential. Belief-at-Risk, developed formally in Section~\ref{sec:bar}, provides this diagnostic by combining posterior uncertainty, belief instability, and expected downside consequence.

The Bayesian filtering framework therefore answers the inference part of the governance problem: it transforms sequential observations into posterior beliefs. The next section answers the decision part by specifying how these evolving beliefs, together with AI recommendations, reference recommendations, and governance risk, are translated into delegated authority through an optimized governance policy.

\section{Governance Policy and Sequential Optimization}
\label{sec:optimal_governance}

Section~\ref{sec:bayesian_pomdp} established how Bayesian filtering transforms heterogeneous observations into posterior beliefs over latent organizational states. Bayesian inference alone, however, does not determine organizational behaviour. The central managerial decision is how much authority should be delegated to AI-generated recommendations once uncertainty has been quantified. This section develops the governance policy that solves this delegation problem.

Unlike a classical POMDP, the optimization variable is not the operational action itself. The operational recommendation is produced independently by the AI system. Instead, the governance policy optimizes the degree of authority assigned to that recommendation relative to an independently validated reference process. Consequently, sequential optimization occurs at the governance layer rather than at the operational layer, separating statistical inference from organizational decision rights.

\subsection{Decision Context}

Governance decisions require substantially more information than the AI recommendation alone. Organizations must evaluate the current posterior belief regarding latent conditions, the AI recommendation, the independently validated reference recommendation, governance uncertainty, empirical reliability, reported confidence, and additional validation diagnostics. Collectively, these variables summarize the information available before organizational authority is allocated.

\begin{equation}
 s^G_t=(b_t,u^{AI}_t,u^{Ref}_t,\BaR_t,c_t,\tau_t,\eta_t).
 \label{eq:governance_state}
\end{equation}

The governance state intentionally separates informational uncertainty from organizational trust. Posterior beliefs summarize Bayesian evidence, Belief-at-Risk summarizes governance uncertainty and potential consequence, reliability summarizes historical validation performance, and confidence records the AI system's own assessment. Because these quantities represent different managerial concepts, they are retained as separate state variables rather than collapsed into a single trust score.

Having defined the information available to governance, the next step is to determine how delegated authority is implemented operationally.

\subsection{Governance Mechanism}

The governance action determines how organizational authority is allocated between AI-generated intelligence and an independently validated reference process. Rather than making a binary approval decision, the policy assigns a continuous delegation weight that blends the two recommendations.

\begin{equation}
 u^{Exec}_t=\alpha(g_t)u^{AI}_t+(1-\alpha(g_t))u^{Ref}_t.
 \label{eq:executed_action}
\end{equation}

This convex combination has several desirable properties. First, it preserves compatibility with existing organizational workflows because the reference process always remains available. Second, it permits gradual adaptation as AI evidence improves rather than requiring abrupt transitions between acceptance and rejection. Finally, it ensures that delegated authority becomes the optimization variable of the governance policy, while operational decision making remains unchanged.

The remaining question is how the organization should determine the appropriate delegation weight. This requires specifying the organizational objective that the governance policy seeks to optimize.

\subsection{Optimization Objective}

The governance objective is derived from a constrained organizational decision problem. At each decision epoch the organization seeks to maximize the expected value generated by the executed recommendation while satisfying constraints on operational risk, governance uncertainty, institutional consistency, and governance intervention.

Formally, consider the constrained optimization problem
\[
\max_{g_t\in\mathcal G}
\mathbb E_t\!\left[R_{t+1}(u_t^{Exec})\right]
\]
subject to acceptable levels of operational risk, delegated authority under uncertain beliefs, deviation from independently validated organizational decisions, and intervention cost. Introducing Lagrange multipliers for these constraints yields the equivalent penalized objective. The resulting governance utility is therefore interpreted as the Lagrangian representation of organizational preferences rather than as an ad hoc reward function.

\begin{equation}
\begin{aligned}
 r^G(s^G_t,g_t)
 &=\mathbb E_t\left[R_{t+1}(u^{Exec}_t)\right]
 -\lambda_{risk}Risk_t(u^{Exec}_t)
 -\lambda_{BaR}\alpha(g_t)\BaR_t \\
 &\quad -\lambda_{dev}\|u^{Exec}_t-u^{Ref}_t\|^2
 -\lambda_{int}C(g_t).
\end{aligned}
\label{eq:governance_utility}
\end{equation}

Each penalty has a direct managerial interpretation. The operational-risk penalty reflects institutional risk appetite. The Belief-at-Risk penalty limits delegated authority when Bayesian uncertainty and organizational consequence are simultaneously elevated. The deviation penalty discourages unnecessary departures from validated organizational knowledge, while the intervention penalty captures the operational cost of review, escalation, and governance oversight.

Importantly, Belief-at-Risk is weighted by the delegation variable itself. Governance uncertainty therefore becomes costly only when authority is actually delegated to the AI system. High uncertainty alone does not reduce organizational value if the governance policy allocates little authority to AI. This distinction captures the central design principle of the proposed framework: uncertainty becomes consequential only through organizational delegation.

The governance utility therefore specifies what constitutes an optimal delegation decision at a single decision epoch. Organizational decision making, however, is inherently sequential because today's delegation decision influences tomorrow's information, validation history, and governance opportunities.

\subsection{Dynamic Decision Problem}

Governance decisions influence future organizational states. Delegating greater authority today affects subsequent observations, validation outcomes, posterior beliefs, and future governance opportunities. The appropriate objective is therefore dynamic rather than myopic. Dynamic programming provides the natural framework for optimizing delegated authority over time.

\begin{equation}
 V(s^G_t)=\max_{g_t\in\mathcal G}\left\{r^G(s^G_t,g_t)+\gamma\mathbb E\left[V(s^G_{t+1})\mid s^G_t,g_t\right]\right\}.
 \label{eq:bellman_governance}
\end{equation}

The Bellman recursion characterizes an optimal governance policy that maximizes expected discounted organizational utility. Unlike conventional POMDP applications, the optimization is performed over governance actions rather than operational actions. Bayesian filtering estimates the informational state, while the Bellman recursion determines how that information should influence delegated authority.

From an organizational perspective, the framework separates three distinct functions. Bayesian inference estimates latent conditions, Belief-at-Risk quantifies governance uncertainty, and dynamic optimization determines the economically appropriate degree of delegation. Adaptive AI delegation therefore emerges as the solution of a sequential organizational optimization problem rather than as a manually specified approval rule or confidence threshold.

The next section returns to the Bayesian diagnostics introduced in Section~\ref{sec:bayesian_pomdp} by developing Belief-at-Risk as the quantitative governance-risk measure that enters the optimization objective and governs the adaptive allocation of decision authority.

\section{Belief-at-Risk as a Governance Diagnostic}
\label{sec:bar}

Effective governance requires more than measuring predictive uncertainty alone. Organizations must also consider whether uncertainty is changing rapidly and whether acting on uncertain information could generate materially adverse consequences. A posterior distribution may be diffuse but stable, or concentrated but moving sharply in a dangerous direction; the governance implications of these cases are different. Belief-at-Risk (BaR) is introduced to summarize these dimensions within a single diagnostic that can be incorporated directly into the governance utility.

\subsection{Uncertainty Diagnostic}

The first component of BaR measures posterior uncertainty through normalized entropy,
\begin{equation}
 H_t=-\frac{\sum_{i=1}^{K} b_t(i)\log b_t(i)}{\log K}.
 \label{eq:entropy}
\end{equation}
Entropy is high when posterior beliefs are dispersed across latent regimes and low when the organization has concentrated evidence. Although uncertainty is informative, it cannot distinguish between a stable uncertain environment and one in which beliefs are changing rapidly. This motivates a second diagnostic.

\subsection{Fragility Assessment}

Belief instability is measured by the Kullback--Leibler divergence $D_t=D_{KL}(b_t\|b_{t-1})$, which increases when the organization's beliefs shift abruptly as new evidence arrives. Large values of $D_t$ indicate increased fragility in the inference process. However, uncertainty and instability alone are insufficient because identical belief dynamics may imply very different organizational outcomes. We therefore introduce a consequence term, defined as $C_t=\sum_i b_t(i)L_i$, where $L_i$ denotes the loss or organizational consequence associated with latent state $i$.

\subsection{Integrated Governance Measure}

The preceding diagnostics are combined into the Belief-at-Risk measure,
\begin{equation}
 \BaR_t=H_t(1+D_t)C_t.
 \label{eq:bar}
\end{equation}
This specification is intentionally simple and interpretable. BaR increases when beliefs are uncertain, unstable, or associated with adverse consequences. The governance policy uses BaR not as a prediction of realized loss, but as a measure of the organizational risk associated with delegating authority to AI-generated recommendations under the current belief state. In this respect, BaR plays a role analogous to risk measures in finance while remaining grounded in Bayesian belief states and governance consequences.

The particular functional form in Equation~\eqref{eq:bar} is not essential to the architecture. Alternative coherent or distributionally robust governance risk measures could be substituted, provided they preserve the interpretation that higher governance uncertainty should increase the cost of delegating authority to AI. The empirical analysis uses this simple specification because it is transparent, monotone in the relevant components, and easy to validate through controlled perturbations.


\section{Structural Properties of the Governance Policy}
\label{sec:theory}

The governance policy should be evaluated not only by realized portfolio outcomes but also by the structural properties it imposes on organizational delegation. These properties clarify what is guaranteed by the formulation before any particular data set is considered. The results below are intentionally stated under transparent assumptions. They do not claim that AI delegation is universally optimal; instead, they establish that the policy is well defined, preserves reference-model fallback, keeps executed decisions admissible, reduces delegation when governance risk increases, and remains invariant to confidence-only perturbations when validated evidence and recommendations are unchanged. This section therefore provides the theoretical bridge between the decision model and the validation experiments.

\begin{assumption}[State, action, and utility primitives]
\label{ass:governance_primitives}
The governance action set $\mathcal G$ is finite and totally ordered by delegated authority: $g_i\preceq g_j$ if and only if $\alpha(g_i)\leq \alpha(g_j)$, where $\alpha(g)\in[0,1]$. The one-period utility $r_G(s,g)$ is bounded and measurable for every admissible state-action pair, the discount factor satisfies $\gamma\in(0,1)$, and the transition kernel $Q(\cdot\mid s,g)$ is well defined on the governance state space. The AI recommendation $u^{AI}$ and the reference recommendation $u^{Ref}$ belong to a convex admissible set $\mathcal U$ in a normed vector space.
\end{assumption}

\begin{theorem}[Existence and uniqueness of the governance value function]
\label{thm:existence}
Under Assumption~\ref{ass:governance_primitives}, the Bellman operator
\[
(\mathcal T V)(s)=\max_{g\in\mathcal G}\left\{r_G(s,g)+\gamma\int V(s')Q(ds'\mid s,g)\right\}
\]
is a contraction on the space of bounded measurable value functions under the supremum norm. Hence there exists a unique bounded fixed point $V^\star=\mathcal T V^\star$. Moreover, because $\mathcal G$ is finite, there exists an optimal stationary governance policy
\[
\pi_G^\star(s)\in\arg\max_{g\in\mathcal G}\left\{r_G(s,g)+\gamma\int V^\star(s')Q(ds'\mid s,g)\right\}.
\]
\end{theorem}
\begin{proof}
Let $V$ and $W$ be bounded measurable functions. For any fixed state $s$, the elementary inequality $|\max_g a_g-\max_g b_g|\leq \max_g |a_g-b_g|$ gives
\[
| (\mathcal T V)(s)-(\mathcal T W)(s)|\leq \gamma \max_{g\in\mathcal G}\left|\int\{V(s')-W(s')\}Q(ds'\mid s,g)\right|.
\]
The absolute value of the integral is bounded by $\|V-W\|_\infty$, so $| (\mathcal T V)(s)-(\mathcal T W)(s)|\leq \gamma\|V-W\|_\infty$ for every $s$. Taking the supremum over $s$ yields $\|\mathcal T V-\mathcal T W\|_\infty\leq\gamma\|V-W\|_\infty$. Since $0<\gamma<1$, Banach's fixed-point theorem implies the existence and uniqueness of a bounded fixed point. The finite action set ensures that the maximum is attained, yielding an optimal stationary policy.\end{proof}
From an organizational perspective, Theorem~\ref{thm:existence} establishes that delegated decision authority is determined by a well-defined dynamic optimization problem rather than by an ad hoc scoring rule. The empirical sections therefore evaluate the behaviour of a policy with a fixed-point foundation in dynamic programming \citep{bellman1957dynamic,bertsekas2017dynamic,powell2022reinforcement}.

\begin{theorem}[Bayesian sufficiency for governance]
\label{thm:bayesian_sufficiency}
Suppose the latent state $X_t$ is Markov, AI evidence is conditionally incorporated through the filtering recursion, and the diagnostics contained in $(\BaR_t,\tau_t,\eta_t)$ are recursively updated from $(b_t,u_t^{AI},u_t^{Ref},c_t)$ and realized outcomes. Then the governance state $s_t^G$ is Markov, and the posterior belief $b_t=P(X_t\mid\mathcal F_t)$ is a sufficient statistic for the latent-state component of governance.
\end{theorem}
\begin{proof}
The Markov property gives $P(X_{t+1}\mid X_t,\mathcal F_t)=P(X_{t+1}\mid X_t)$. Conditional on the filtration, the posterior belief $b_t$ is the conditional distribution of $X_t$. For any bounded measurable function $h$,
\[
\E[h(X_{t+1})\mid\mathcal F_t]=\sum_i b_t(i)\sum_j P_{ij}h(j),
\]
so the predictive distribution of the latent state depends on the history only through $b_t$. Since the remaining diagnostics are updated recursively by assumption, the conditional distribution of the next governance state depends on the past only through $s_t^G$ and the selected governance action. Hence $s_t^G$ is Markov and $b_t$ is sufficient for latent-state prediction.\end{proof}
This result formalizes the separation between statistical inference and organizational action. Organizations need not carry the full observation history into the governance decision because posterior beliefs summarize the information required for sequential delegation.

\begin{proposition}[Reference consistency and admissibility]
\label{prop:reference_consistency}
For any governance action $g$, the executed recommendation satisfies
\[
u^{Exec}=\alpha(g)u^{AI}+(1-\alpha(g))u^{Ref}.
\]
If $\alpha(g)=0$, then $u^{Exec}=u^{Ref}$; if $\alpha(g)=1$, then $u^{Exec}=u^{AI}$. If $u^{AI},u^{Ref}\in\mathcal U$ and $\mathcal U$ is convex, then $u^{Exec}\in\mathcal U$ for every $g\in\mathcal G$.
\end{proposition}
\begin{proof}
The limiting cases follow by direct substitution. For admissibility, $\alpha(g)\in[0,1]$ by Assumption~\ref{ass:governance_primitives}, so $u^{Exec}$ is a convex combination of two elements of $\mathcal U$. Convexity of $\mathcal U$ implies $u^{Exec}\in\mathcal U$.\end{proof}
This property is managerially important because the governance policy never creates an uncontrolled action outside the approved decision set. It reallocates authority between an AI recommendation and a reference process while preserving the feasibility constraints of the organization.

\begin{theorem}[Belief-at-Risk monotonicity]
\label{thm:bar_monotonicity}
Fix all state components except $\BaR$ and suppose that the Bellman objective has decreasing differences in $(\BaR,g)$ under the order induced by $\alpha(g)$. In particular, suppose the direct BaR penalty is $-\lambda_{BaR}\alpha(g)\BaR$ with $\lambda_{BaR}\geq0$ and the continuation term does not introduce offsetting increasing differences. Then the optimal delegation correspondence is nonincreasing in $\BaR$.
\end{theorem}
\begin{proof}
For two actions $g_i\preceq g_j$, define $\Delta(\BaR)=Q(g_j,\BaR)-Q(g_i,\BaR)$, where $Q$ is the Bellman objective. The BaR component contributes
\[
-\lambda_{BaR}\{\alpha(g_j)-\alpha(g_i)\}\BaR,
\]
which is nonincreasing in $\BaR$ because $\lambda_{BaR}\geq0$ and $\alpha(g_j)\geq\alpha(g_i)$. By assumption, the remaining components do not reverse this decreasing-differences property. Therefore increases in $\BaR$ weakly reduce the relative attractiveness of higher-delegation actions. On a finite ordered action set, decreasing differences imply that the argmax correspondence cannot shift upward as $\BaR$ increases.\end{proof}
The theorem gives a formal statement of governance conservatism: rising governance uncertainty should not increase delegated authority to AI-generated recommendations. This is the theoretical basis for the BaR sweep and governance-appetite validation.

\begin{theorem}[Reported-confidence robustness]
\label{thm:confidence_robustness}
Suppose a perturbation changes only reported confidence $c_t$ while leaving $b_t$, $u_t^{AI}$, $u_t^{Ref}$, $\BaR_t$, $\tau_t$, $\eta_t$, the reward function, and the transition kernel unchanged. Then the optimal governance action correspondence and the set of optimal delegation weights are unchanged.
\end{theorem}
\begin{proof}
Under the stated perturbation, every term in the Bellman objective is identical for each action $g\in\mathcal G$. The one-period utility is unchanged because the variables that enter reward, risk, BaR, deviation, and intervention costs are fixed. The continuation value is unchanged because the transition kernel is fixed. Thus the vector of action values is identical before and after the confidence perturbation. Identical action-value vectors have identical argmax correspondences, and therefore identical optimal delegation-weight sets.\end{proof}
This theorem does not claim that confidence is always irrelevant. If reported confidence changes beliefs, recommendations, reliability diagnostics, or transition assessments, it may affect governance. The result establishes the narrower and managerially important principle that confidence declarations alone should not mechanically alter organizational trust.

\begin{theorem}[Adaptive delegation under improving evidence]
\label{thm:adaptive_delegation}
Consider two evidence environments $e_0$ and $e_1$ with the same reported confidence and intervention costs. If moving from $e_0$ to $e_1$ weakly increases the incremental value of higher-delegation actions relative to lower-delegation actions, then the optimal delegation correspondence under $e_1$ is weakly higher than under $e_0$.
\end{theorem}
\begin{proof}
Let $Q_e(g)$ denote the Bellman objective under evidence environment $e$. The assumption states that for any $g_i\preceq g_j$, $Q_{e_1}(g_j)-Q_{e_1}(g_i)\geq Q_{e_0}(g_j)-Q_{e_0}(g_i)$. This is increasing differences in $(e,g)$ on a finite ordered action set. Standard monotone-choice reasoning \citep{topkis1998supermodularity} implies that the set of maximizers under $e_1$ is weakly higher than the set of maximizers under $e_0$. Therefore delegated authority increases weakly when evidence quality improves in the stated sense.\end{proof}
The theorem formalizes the adaptive-delegation intuition tested in the forecast-accuracy validation: when evidence becomes more informative while reported confidence is held fixed, organizations should allocate more authority to AI-generated recommendations.

\begin{table}[H]
\centering
\small
\caption{Roadmap linking structural properties to validation experiments. The table shows how each formal property is paired with a corresponding empirical test, ensuring that the theoretical analysis and validation methodology evaluate the same governance behaviour.}
\label{tab:validation_roadmap}
\resizebox{\textwidth}{!}{%
\begin{tabular}{lll}
\toprule
Structural property & Managerial interpretation & Validation evidence \\
\midrule
Existence of optimal policy & Delegation is a well-defined dynamic optimization problem & Bellman implementation \\
Bayesian sufficiency & Decisions depend on posterior information, not raw histories & Belief-state recovery \\
Reference consistency and admissibility & Organizations retain a feasible reference fallback & Stress tests and bad-AI scenarios \\
BaR monotonicity & Delegation declines as governance uncertainty rises & BaR sweep and appetite sensitivity \\
Reported-confidence robustness & Confidence is not organizational trust & Confidence-only perturbation \\
Adaptive delegation & Authority rises with Bayesian evidence quality & Forecast-accuracy validation \\
Governance calibration & Risk appetite is represented by policy parameters & Governance-appetite sensitivity \\
Early warning & Governance risk rises before adverse outcomes & Fragile-AI early-warning experiment \\
\bottomrule
\end{tabular}%
}
\end{table}

\section{Parameterization and Governance Appetite}
\label{sec:parameterization}

The governance policy contains parameters that translate organizational objectives into quantitative delegation decisions. These parameters should not be interpreted as arbitrary tuning constants. They represent managerial choices about risk appetite, tolerance for deviation from independently validated processes, and the operational cost of governance intervention. This distinction is important because organizations differ in regulatory exposure, strategic objectives, and willingness to delegate authority to AI. A useful governance framework should therefore separate model design from policy calibration.

The most important parameter in the empirical analysis is the Belief-at-Risk penalty $\lambda_{BaR}$. Increasing this parameter raises the cost of delegating authority to AI under uncertain, unstable, or adverse posterior beliefs. A conservative organization can therefore reduce AI delegation by increasing $\lambda_{BaR}$ without changing the AI model, the reference process, or the Bayesian filter. Conversely, an organization seeking greater use of AI-generated recommendations may reduce this penalty while preserving the same governance architecture. This is managerially useful because changes in risk appetite can be implemented through transparent policy parameters rather than through ad hoc changes to approval workflows.

The remaining parameters have analogous interpretations. The risk penalty $\lambda_{risk}$ reflects the organization's aversion to financial or operational risk. The deviation penalty $\lambda_{dev}$ reflects the cost of departing from independently validated reference decisions. The intervention penalty $\lambda_{int}$ reflects the operational burden associated with review, escalation, monitoring, or rejection. Together these parameters provide a compact representation of governance appetite that can be reviewed, audited, and calibrated over time. The validation experiments in later sections examine whether changing these parameters produces interpretable changes in delegated authority, which is a necessary property for practical governance deployment.

\section{Validation Framework}
\label{sec:experimental_design}

The objective of the empirical analysis is not to determine whether a particular AI model outperforms alternative forecasting methods, nor to identify a portfolio strategy with superior historical performance. The experiments answer a different question at the heart of organizational AI governance: does the governance policy exhibit the structural properties required for dynamically allocating decision authority under uncertainty? This distinction is important because governance is a decision policy rather than a forecasting model. The empirical objective is therefore to validate the behaviour of the governance policy itself: how it allocates authority, responds to uncertainty, preserves a reference fallback, and adapts when evidence quality changes.

The validation framework mirrors the theoretical development in Sections~\ref{sec:design_principles}--\ref{sec:theory}. Each experiment isolates one structural property while holding the remaining components approximately fixed. Rather than evaluating governance solely through aggregate performance statistics, the experiments examine whether the policy behaves consistently with the decision-theoretic principles on which it is constructed. In particular, the experiments investigate whether governance degrades gracefully when AI behaviour deteriorates, whether delegation remains robust to confidence-only perturbations, whether delegated authority responds appropriately to changes in Bayesian evidence quality, whether institutional conservatism can be calibrated continuously through the governance objective, and whether the proposed sequential policy adds value relative to simpler delegation mechanisms.

The empirical design is organized around five complementary tests. The first four tests validate structural properties of the proposed Governance-Aware POMDP. The fifth test benchmarks the proposed policy against simpler governance strategies under a common information set and a common governance objective. This distinction is important: the structural experiments ask whether the proposed policy behaves as theory predicts, while the benchmark asks whether sequential Bayesian governance provides organizational value relative to alternative governance rules that an organization might plausibly implement.

\subsection{Stress Robustness}
\label{subsec:stress_robustness}

The first experiment evaluates the governance policy across four representative AI operating conditions: Low Reliability, Noisy, Improving, and High Reliability. These scenarios are not equally spaced points on a one-dimensional quality scale. They capture qualitatively different forms of AI behaviour that organizations may encounter during deployment, including degraded probabilistic evidence, unstable recommendations, improving performance through learning, and consistently reliable AI. The primary diagnostics are adaptive delegation weight, AI rejection rate, Belief-at-Risk, and changes in portfolio outcomes relative to both the reference process and the unguided AI system. From a managerial perspective, the experiment asks whether governance continues to function appropriately when AI systems become unreliable rather than assuming that AI capability remains fixed.

\subsection{Reported LLM-Confidence Robustness}
\label{subsec:reported_confidence_robustness}

Many AI systems produce confidence estimates together with recommended actions. Such confidence values are often interpreted as indicators of reliability, even though they need not be statistically calibrated. The second experiment therefore investigates whether governance responds directly to reported confidence or whether delegation is determined primarily by externally validated Bayesian evidence. To isolate confidence, the reported confidence field is varied continuously while Bayesian beliefs, AI recommendations, and Belief-at-Risk remain approximately unchanged. If governance is functioning as intended, confidence-only perturbations should produce little or no change in delegated authority because confidence is treated as supplementary information rather than as a direct measure of organizational trust.

\subsection{Forecast Accuracy Validation}
\label{subsec:forecast_accuracy_validation}

The third experiment isolates Bayesian evidence quality while approximately holding reported confidence constant. AI evidence is progressively interpolated between degraded and original belief vectors, and AI recommendations are perturbed analogously. This design separates informational quality from subjective confidence and evaluates whether governance responds to the quantity that should matter organizationally: the quality of the available evidence. Adaptive delegation is quantified using delegated authority, governance elasticity, regression slope, coefficient of determination, and rank correlation. Unlike the stress tests, this experiment provides a continuous characterization of governance behaviour as evidence quality improves.

\subsection{Governance Appetite Sensitivity}
\label{subsec:governance_appetite_sensitivity}

Organizations differ substantially in their tolerance for uncertainty and operational risk. Rather than requiring different governance architectures for different institutions, the proposed framework represents organizational conservatism through the Belief-at-Risk penalty in the governance utility. The calibration experiment varies $\lambda_{BaR}$ while holding the AI outputs and market observations fixed. Increasing this penalty should reduce delegated authority and increase intervention, while reducing it should permit greater reliance on AI-generated recommendations when evidence quality is sufficient. This experiment demonstrates that governance appetite can be calibrated continuously without redesigning the policy itself.

\subsection{Benchmark Governance Policies}

The preceding experiments evaluate whether the proposed governance policy satisfies its intended structural properties. A complementary question is whether sequential Bayesian governance adds value relative to simpler governance mechanisms that organizations might plausibly deploy. The benchmark experiment therefore compares the Governance-Aware POMDP with five alternative policies: Static Delegation, Confidence Threshold, Reliability-Only Delegation, Bayesian Shrinkage, and SR11-7 Style Governance.

The benchmark design deliberately holds the informational environment fixed. All policies observe the same simulated market paths, AI evidence, Bayesian posterior beliefs, reliability updates, Belief-at-Risk diagnostics, and reference recommendations. The policies differ only in the rule used to select delegated AI decision exposure. After the exposure decision is chosen, every policy is evaluated using the same common governance objective. This design ensures that performance differences reflect governance-policy structure rather than differences in information, filtering, or utility specification.

The benchmark is evaluated using two complementary criteria. Investment performance is summarized using governed Sharpe ratio and drawdown. Governance performance is summarized using the common governance utility, which rewards validated AI value but penalizes delegated AI authority under elevated model risk, adverse-state probability, switching costs, and intervention costs. This separation is important because the proposed framework is a governance mechanism, not a portfolio optimizer. A simple policy may perform well in one stationary scenario, but the central question is whether it provides robust decision-authority allocation across heterogeneous AI-quality regimes.

The benchmark also clarifies the role of specialization. Bayesian Shrinkage is expected to perform well when AI quality is persistently poor because aggressive contraction toward the reference process is appropriate in that limiting case. Confidence-threshold and reliability-only policies are expected to perform well when their triggering variables are well calibrated. Static delegation provides a simple nonadaptive baseline. The Governance-Aware POMDP is expected to be most valuable when AI quality is uncertain, unstable, or evolving, because it uses sequential Bayesian governance rather than a fixed rule for delegated authority. The benchmark therefore evaluates not only whether the proposed policy performs well, but when its additional sequential structure is managerially justified.

Together, these experiments validate governance as an organizational decision policy rather than evaluating a prediction model alone. Stress tests examine robustness across representative operating conditions; the reported-confidence experiment distinguishes organizational trust from AI confidence; the forecast-accuracy experiment evaluates adaptive delegation under changing evidence quality; the governance-appetite experiment demonstrates continuous calibration of institutional conservatism; and the benchmark comparison evaluates whether sequential Bayesian governance provides robust general-purpose performance relative to simpler delegation strategies.

\section{Synthetic Governance Laboratory}
\label{sec:synthetic_validation}

The synthetic governance laboratory provides a controlled environment in which the structural properties of the governance policy can be evaluated under known latent regimes and deliberately perturbed AI behaviour. This setting is useful because it allows the experiments to distinguish behaviour induced by the governance policy from behaviour induced by uncontrolled features of historical data. The objective is not to demonstrate a trading strategy but to evaluate whether the policy reallocates delegated authority in ways that are mathematically and managerially coherent. The section proceeds in two stages. First, it validates the structural properties of adaptive delegation under controlled perturbations. Second, it compares the Governance-Aware POMDP with simpler governance policies under identical information, beliefs, simulated paths, and governance objectives.

\subsection{Stress Robustness}
The stress-test experiments evaluate governance across representative AI operating conditions. Figure~\ref{fig:adaptive_trust_allocation} reports delegated authority, rejection, Belief-at-Risk, and performance diagnostics across low-reliability, noisy, improving, and high-reliability scenarios. The governance policy allocates lower authority to AI when behaviour deteriorates and increases reliance as evidence quality improves. Intervention rates rise under degraded AI conditions and decline when AI evidence becomes more reliable. From an organizational perspective, this indicates that governance can protect against deteriorating AI behaviour without permanently suppressing the value of AI-generated intelligence.

\begin{figure}[H]
\centering
\IfFileExists{figures/adaptive_trust_allocation_five_panel.png}{\includegraphics[width=0.95\textwidth]{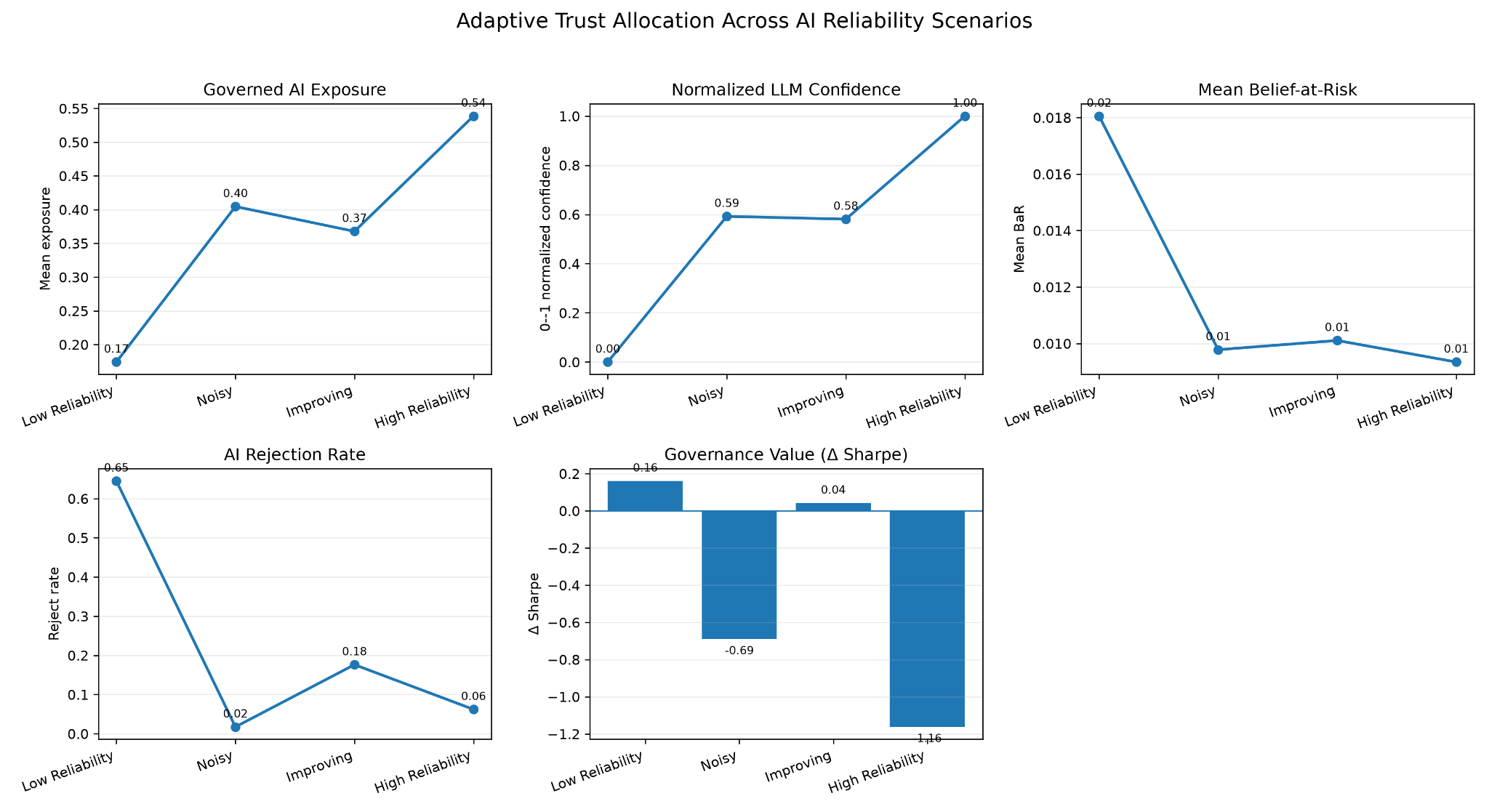}}{\fbox{Missing figure: adaptive trust allocation.}}
\caption{Adaptive delegation across representative AI operating conditions. The figure summarizes how delegated decision authority, rejection, Belief-at-Risk, and performance diagnostics change across degraded, noisy, improving, and reliable AI scenarios. The main empirical message is not that every scenario lies on a single monotone quality curve, but that the governance policy reduces delegated authority when evidence deteriorates while preserving the ability to increase delegation when AI-generated evidence becomes more reliable.}
\label{fig:adaptive_trust_allocation}
\end{figure}

\subsection{Reported LLM-Confidence Robustness}
The reported-confidence robustness experiment isolates the effect of the AI's self-reported confidence while holding posterior beliefs and recommendations fixed. Figure~\ref{fig:confidence_robustness} shows that delegated authority, rejection, and BaR remain effectively unchanged as reported confidence varies. This result validates the proposition that confidence is not organizational trust. The governance policy does not mechanically follow confidence declarations because confidence does not alter Bayesian evidence, recommendations, or governance risk in this experiment.

\begin{figure}[H]
\centering
\IfFileExists{figures/reported_confidence_robustness.png}{\includegraphics[width=0.94\textwidth]{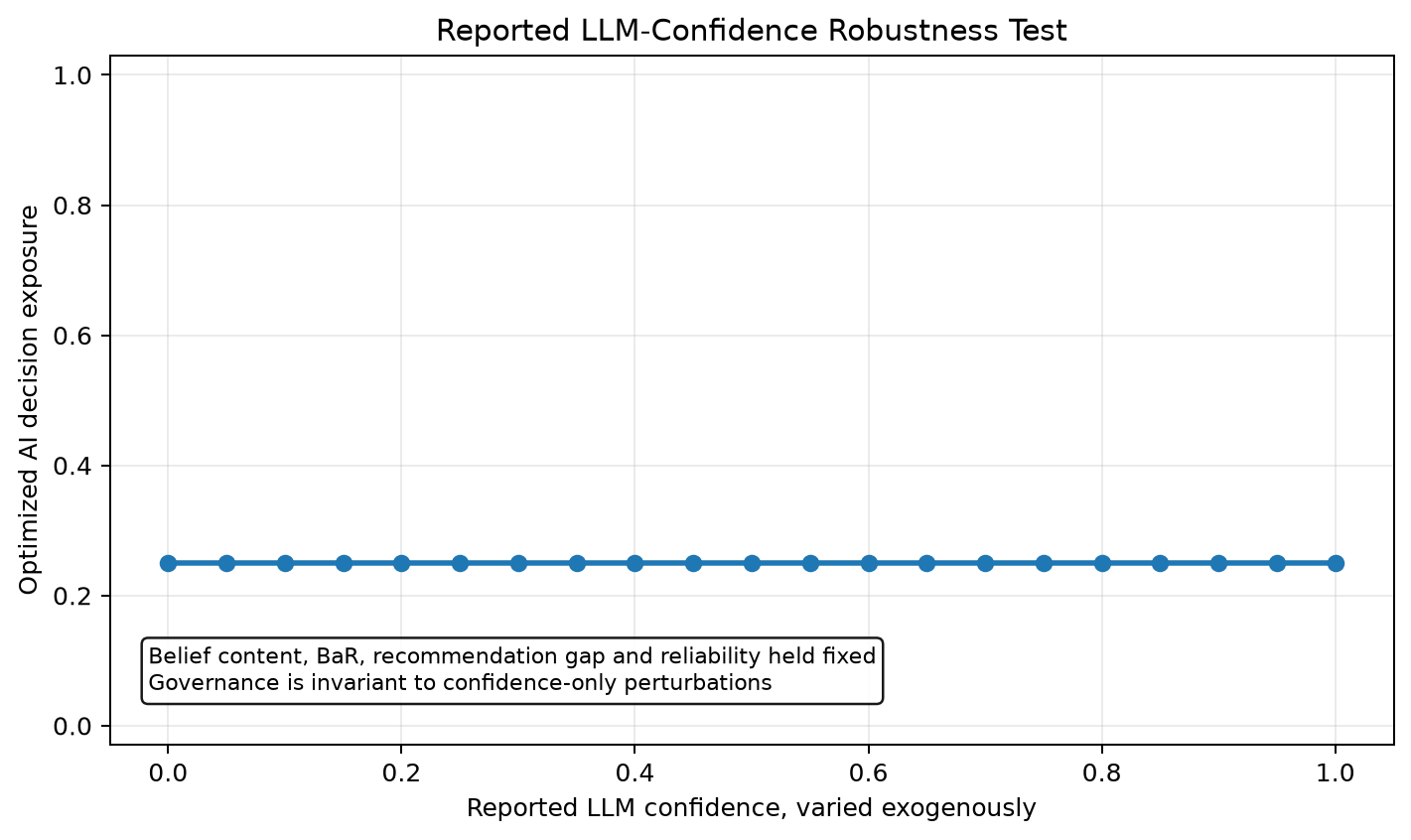}}{\fbox{Missing reported confidence robustness figure.}}
\caption{Reported LLM-confidence robustness. Bayesian beliefs, AI recommendations, and Belief-at-Risk are held fixed while reported confidence is perturbed. The approximately constant governance response demonstrates that delegated authority is not mechanically determined by self-reported AI confidence, which is important for organizations using models whose confidence estimates may be poorly calibrated.}
\label{fig:confidence_robustness}
\end{figure}

\subsection{Forecast Accuracy Validation}
The forecast-accuracy validation experiment improves the content of AI-generated evidence while holding reported confidence approximately fixed. Figure~\ref{fig:forecast_accuracy} shows that delegated authority increases as evidence quality improves, while Belief-at-Risk decreases. This experiment provides the primary continuous validation of adaptive delegation: the organization assigns more authority to AI because evidence becomes more informative, not because the AI becomes more confident in its own output.

\begin{figure}[H]
\centering
\IfFileExists{figures/forecast_accuracy_validation.png}{\includegraphics[width=0.94\textwidth]{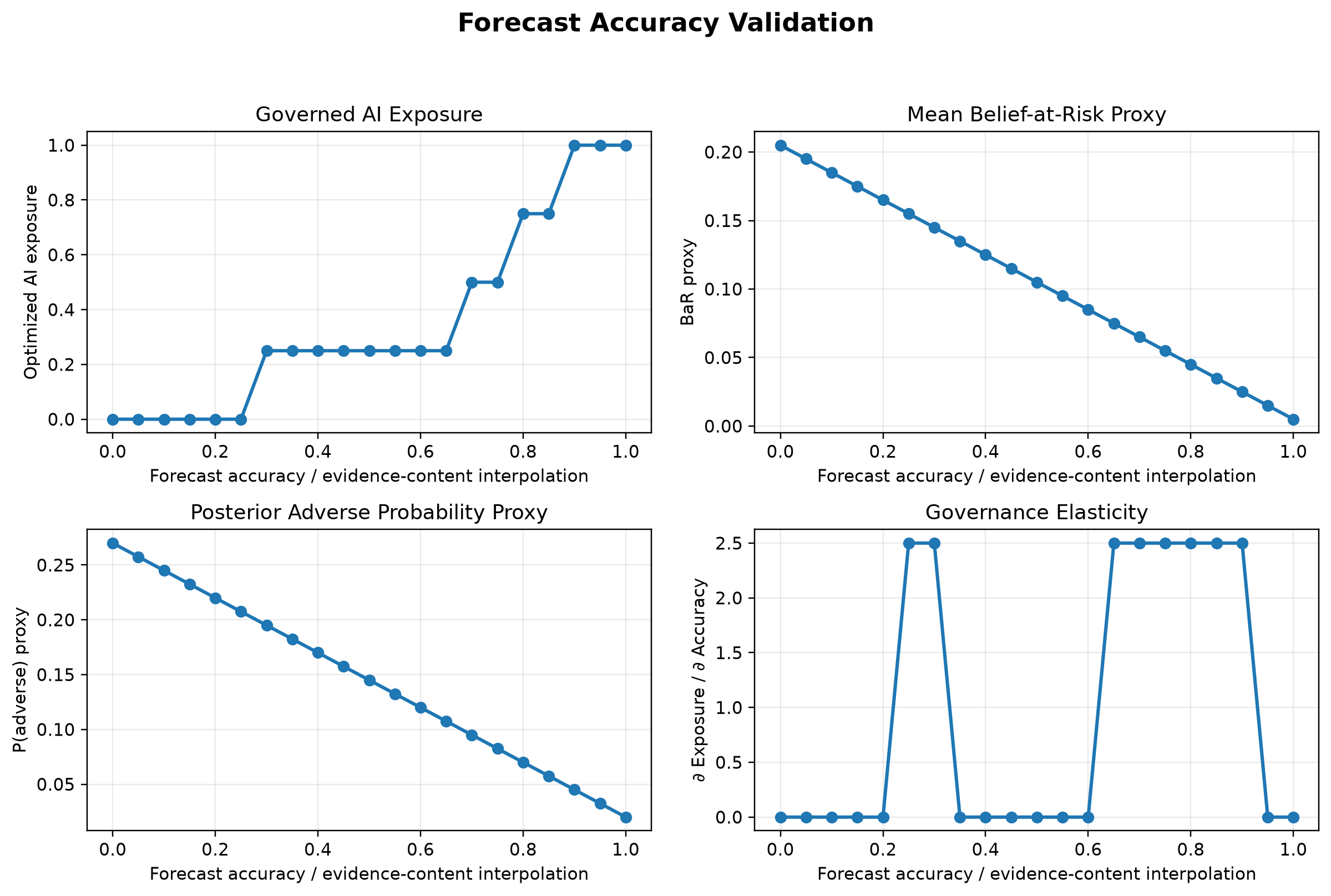}}{\fbox{Missing forecast accuracy validation figure.}}
\caption{Forecast-accuracy validation. AI evidence quality and recommendation content are progressively improved while reported confidence is held approximately constant. Delegated decision authority rises as Bayesian evidence becomes more informative, providing empirical support for adaptive delegation based on evidence quality rather than confidence declarations.}
\label{fig:forecast_accuracy}
\end{figure}

\begin{figure}[H]
\centering
\IfFileExists{figures/reported_confidence_vs_accuracy_validation.png}{\includegraphics[width=0.95\textwidth]{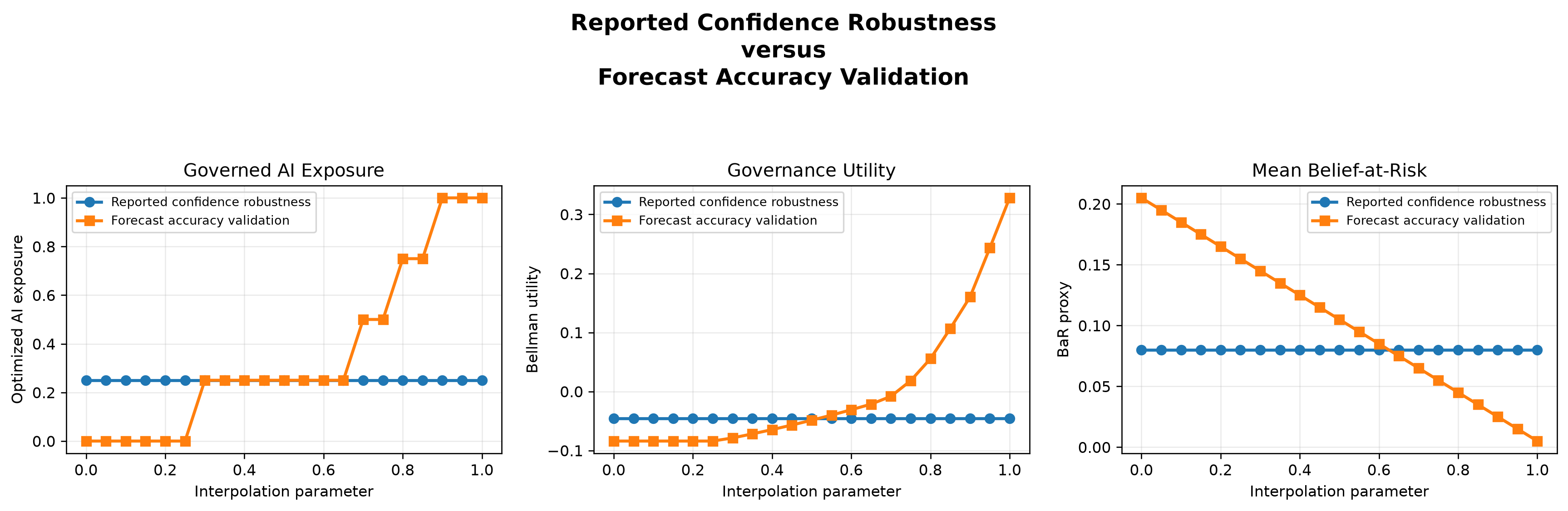}}{\fbox{Missing comparison figure.}}
\caption{Reported confidence robustness versus forecast-accuracy validation. Confidence-only perturbations leave the governance policy largely unchanged, whereas perturbations to Bayesian evidence quality produce systematic changes in delegated authority, rejection, and Belief-at-Risk. The contrast distinguishes AI confidence from organizational trust and supports the design principle that governance should rely primarily on validated evidence quality.}
\label{fig:confidence_vs_accuracy}
\end{figure}

\subsection{Governance Appetite Sensitivity}
The governance-appetite sensitivity experiment varies the Belief-at-Risk penalty while holding the rest of the model fixed. Figure~\ref{fig:governance_appetite} shows that increasing this penalty reduces delegated authority and produces more conservative governance behaviour. The result is managerially important because it demonstrates that institutional conservatism can be calibrated through an interpretable policy parameter rather than through a collection of manually specified approval rules.

\begin{figure}[H]
\centering
\IfFileExists{figures/governance_appetite_sensitivity.png}{\includegraphics[width=0.94\textwidth]{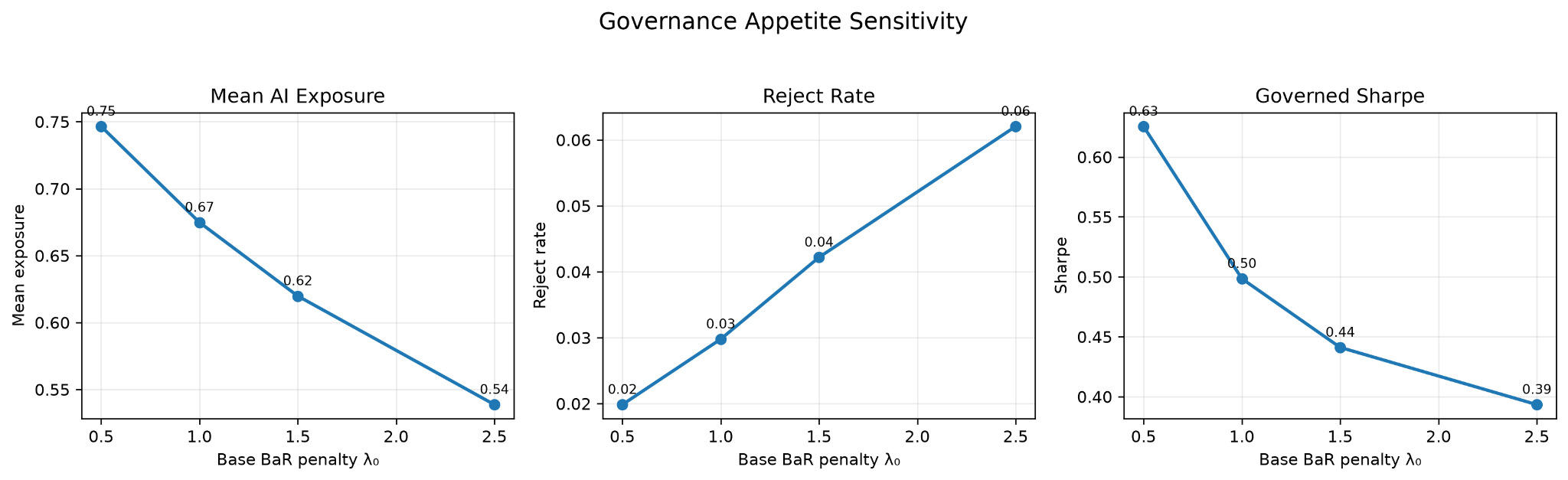}}{\fbox{Missing governance appetite sensitivity figure.}}
\caption{Governance-appetite sensitivity. Increasing the Belief-at-Risk penalty reduces delegated authority and increases governance conservatism without changing the underlying AI outputs or model structure. The figure illustrates how institutions can calibrate AI delegation through a transparent governance appetite parameter.}
\label{fig:governance_appetite}
\end{figure}

\subsection{Fragile AI Early Warning}
The fragile-AI experiment is included in the main text because it visually connects the full governance mechanism. Figure~\ref{fig:fragile_early_warning} shows hidden regimes, statistical reliability, posterior adverse-state probability, Belief-at-Risk, delegated authority, and portfolio drawdowns on a common time axis. The figure illustrates why governance requires more than a performance metric: reliability, beliefs, and BaR begin to deteriorate before the worst portfolio losses emerge, while delegated authority declines as governance risk rises. This provides a concrete example of the governance policy functioning as an early-warning and authority-allocation mechanism rather than merely as an ex post performance filter.

\begin{figure}[H]
\centering
\IfFileExists{figures/fragile_ai_failure_04_early_warning_with_bar.png}{\includegraphics[width=0.99\textwidth]{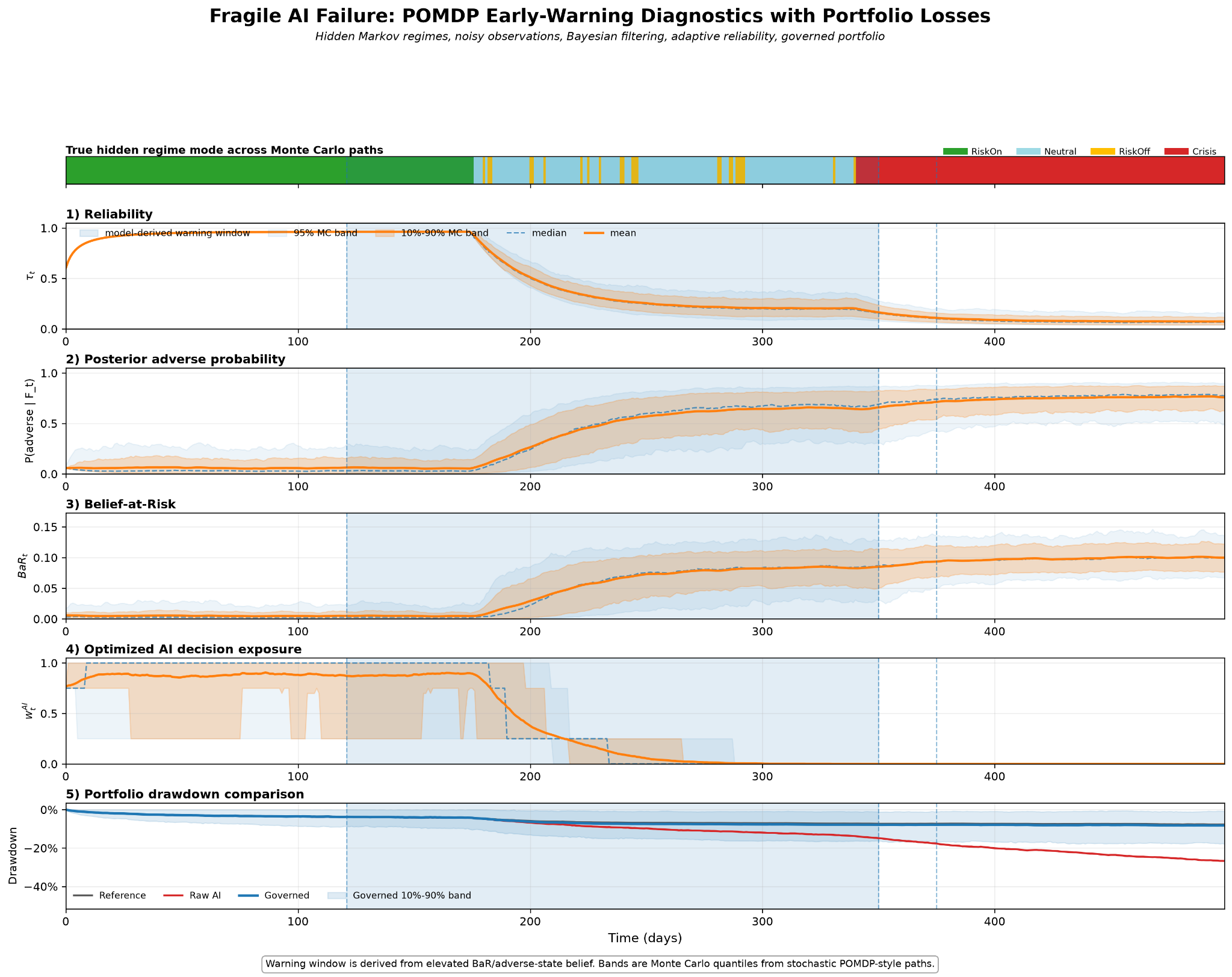}}{\IfFileExists{figures/fragile_ai_failure_04_early_warning.png}{\includegraphics[width=0.99\textwidth]{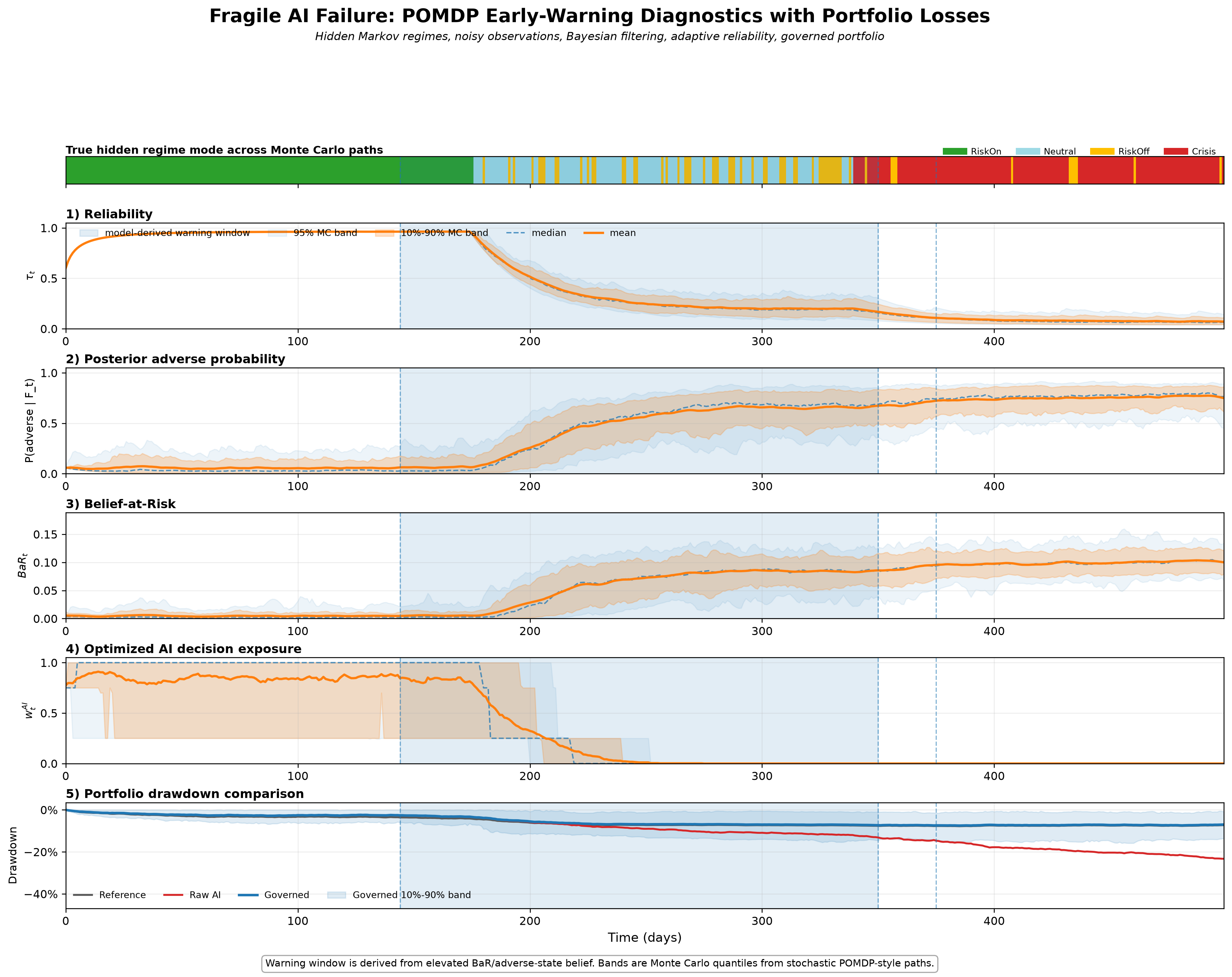}}{\fbox{Missing fragile AI early-warning figure.}}}
\caption{Fragile-AI early-warning diagnostics. The figure links latent regimes, statistical reliability, posterior adverse probability, Belief-at-Risk, optimized delegated authority, and portfolio drawdowns. It demonstrates how the governance policy reduces AI delegation as Bayesian and validation diagnostics deteriorate, providing an interpretable early-warning view of fragile AI behaviour before adverse outcomes fully materialize.}
\label{fig:fragile_early_warning}
\end{figure}

\subsection{Benchmark Governance Policies}
\label{subsec:benchmark_governance_policies}

The preceding experiments validate the structural behaviour of the Governance-Aware POMDP. We next ask whether sequential Bayesian governance provides advantages over simpler governance mechanisms that an organization might plausibly implement. The benchmark compares the proposed policy with five alternatives: Static Delegation, Confidence Threshold, Reliability-Only Delegation, Bayesian Shrinkage, and SR11-7 Style Governance. All policies use the same Bayesian filter, the same simulated paths, the same reference process, and the same governance objective. They differ only in how delegated AI authority is selected.

This design makes the benchmark a governance comparison rather than a comparison of information sets or forecasting models. Each policy receives the same posterior beliefs, reliability estimates, Belief-at-Risk diagnostics, and recommendation streams. Every policy is then scored using the same common governance utility. The benchmark therefore isolates the incremental value of sequential governance relative to static or heuristic allocation rules.

\begin{figure}[H]
\centering
\IfFileExists{figures/benchmark_policy_comparison.png}{\includegraphics[width=0.98\textwidth]{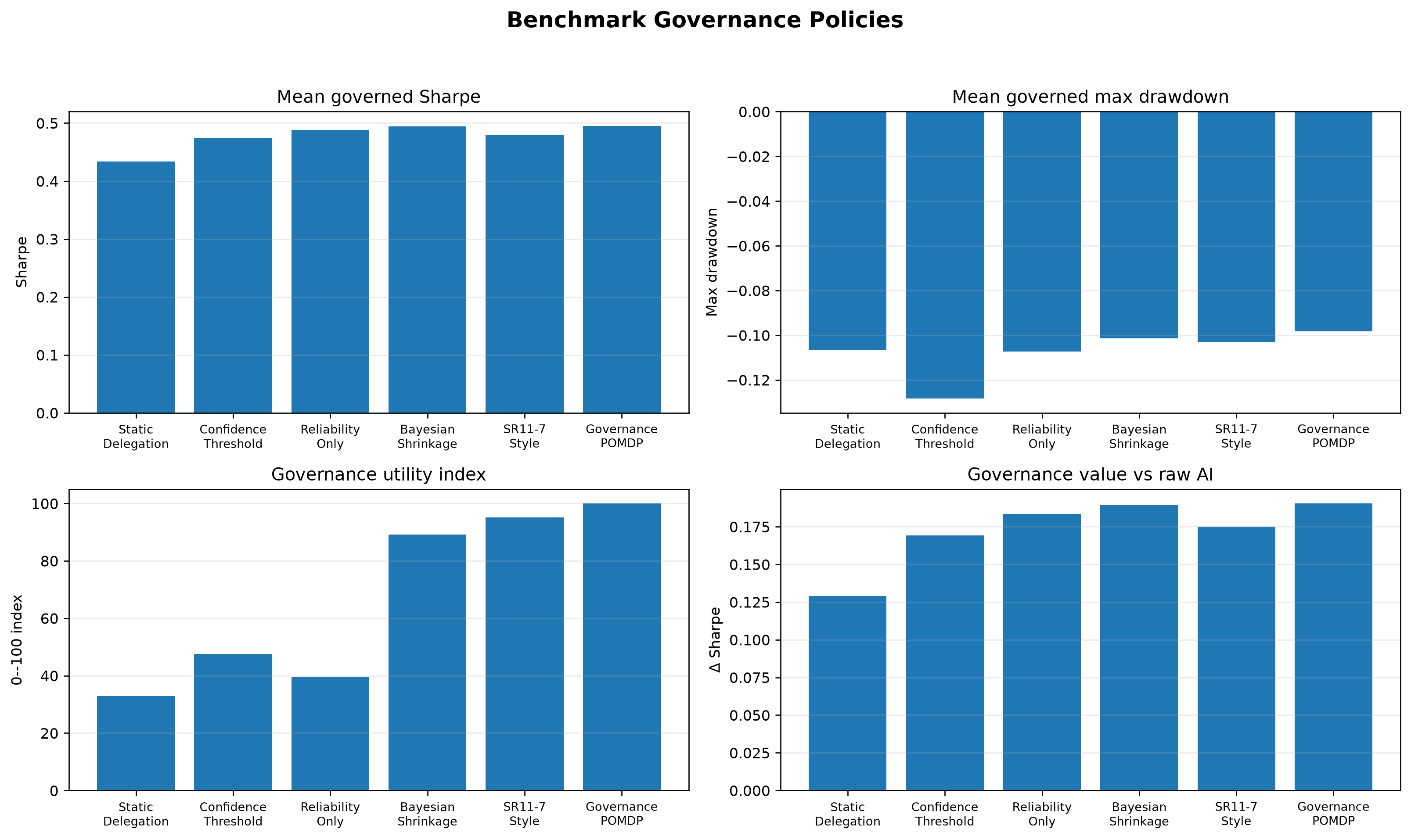}}{\fbox{Missing benchmark policy comparison figure.}}
\caption{Benchmark comparison across governance policies. The figure reports average governed performance across the benchmark policies under the same Bayesian filtering process, reference model, simulated market paths, and governance objective. The comparison separates investment outcomes from governance-policy structure and shows that the Governance-Aware POMDP remains competitive while optimizing the organizational governance objective.}
\label{fig:benchmark_policy_comparison}
\end{figure}

\subsection{Policy Dominance across AI-Quality Regimes}
\label{subsec:benchmark_policy_dominance}

Table~\ref{tab:benchmark_policy_dominance} reports the policy achieving the highest average governance utility in each AI-quality regime, together with the runner-up. The results show that no single simple heuristic dominates across all environments. Bayesian Shrinkage performs best when AI quality is persistently poor because rapid contraction toward the validated reference process minimizes unnecessary AI exposure. This is the regime in which a specialized shrinkage rule is expected to perform well.

In contrast, the Governance-Aware POMDP achieves the highest governance utility in the fragile, improving, and high-quality AI regimes. These regimes require adaptation because AI reliability is unstable, evolving, or valuable but still subject to governance risk. The results therefore distinguish specialized policies from adaptive governance. Simple shrinkage is effective under stationary poor AI quality, whereas sequential Bayesian governance provides the strongest general-purpose policy across dynamic operating environments.


\begin{table}[htbp]
\centering
\caption{Dominant governance policies under different AI-quality regimes. The table reports the policy achieving the highest average governance utility together with the second-best policy. Results are averaged over 100 Monte Carlo simulations per scenario.}
\label{tab:benchmark_policy_dominance}

\resizebox{0.92\textwidth}{!}{%
\begin{tabular}{llclc}
\toprule
\textbf{AI Quality Regime} &
\textbf{Best Policy} &
\textbf{Governance Utility} &
\textbf{Runner-up} &
\textbf{Governance Utility} \\
\midrule

Persistently Poor AI &
Bayesian Shrinkage &
\textbf{-0.0016} &
Governance POMDP &
-0.0022 \\

Fragile / Unstable AI &
\textbf{Governance POMDP} &
\textbf{0.0482} &
SR11-7 Style &
0.0389 \\

Improving AI &
\textbf{Governance POMDP} &
\textbf{0.0557} &
SR11-7 Style &
0.0077 \\

High-Quality AI &
\textbf{Governance POMDP} &
\textbf{0.0472} &
SR11-7 Style &
-0.0042 \\

\bottomrule
\end{tabular}%
}
\end{table}

\subsection{Cross-Scenario Benchmark Rankings}
\label{subsec:benchmark_policy_rankings}

Figure~\ref{fig:benchmark_policy_rankings} and Table~\ref{tab:benchmark_rankings} summarize benchmark performance across all four AI-quality regimes. The Governance-Aware POMDP achieves the best overall ranking by combining competitive investment performance with the strongest governance-utility ranking. The policy is not tuned to maximize Sharpe ratio in every environment. Instead, it optimizes delegated authority by balancing expected AI value against governance risk, intervention costs, and switching costs.

This distinction is managerially important. Confidence Threshold and Reliability-Only policies can perform well in selected investment scenarios, but they may delegate excessive authority when confidence or historical reliability is misleading. Static Delegation provides a simple baseline but cannot adapt to changes in evidence quality. Bayesian Shrinkage is effective when AI is persistently poor, but it does not provide the same general-purpose adaptation when AI reliability evolves. The Governance-Aware POMDP therefore functions as an adaptive decision-authority allocation mechanism rather than as a rule for maximizing AI utilization.

\begin{figure}[H]
\centering
\IfFileExists{figures/benchmark_policy_rankings.png}{\includegraphics[width=0.86\textwidth]{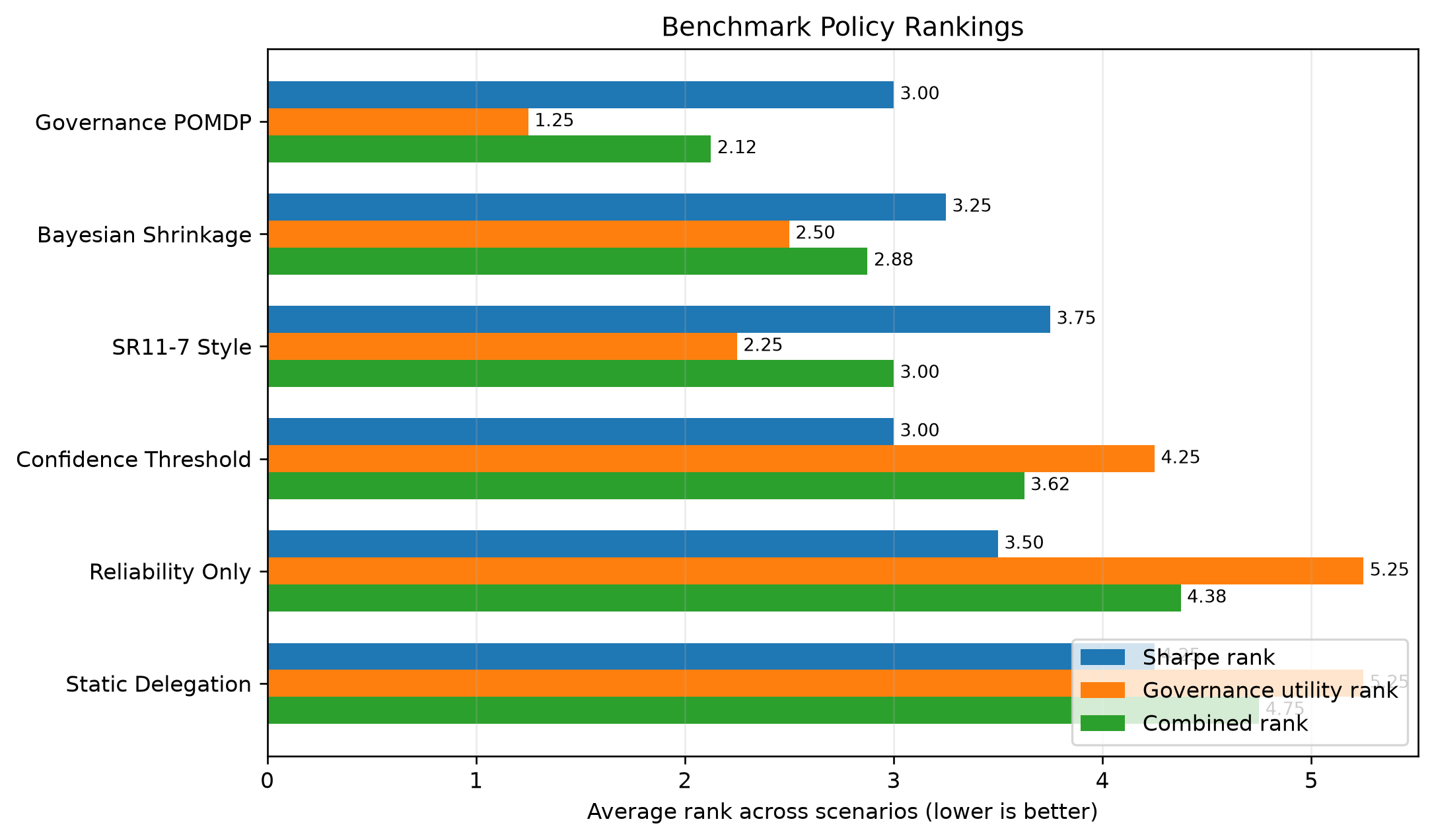}}{\fbox{Missing benchmark policy rankings figure.}}
\caption{Average benchmark rankings across AI-quality regimes. Lower values indicate better performance. The Governance-Aware POMDP achieves the best overall rank by combining competitive investment performance with the strongest average governance-utility rank across heterogeneous AI operating conditions.}
\label{fig:benchmark_policy_rankings}
\end{figure}


\begin{table}[htbp]
\centering
\caption{Average benchmark rankings across the four AI-quality scenarios. Lower ranks indicate better performance. Sharpe rank measures investment performance, governance rank measures performance under the common governance objective, and the overall rank is the average of the two.}
\label{tab:benchmark_rankings}

\resizebox{0.82\textwidth}{!}{%
\begin{tabular}{lccc}
\toprule
\textbf{Governance Policy} &
\textbf{Average Sharpe Rank} &
\textbf{Average Governance Rank} &
\textbf{Overall Rank} \\
\midrule
\textbf{Governance POMDP} & \textbf{3.00} & \textbf{1.25} & \textbf{2.125} \\
Bayesian Shrinkage        & 3.25 & 2.50 & 2.875 \\
SR11-7 Style              & 3.75 & 2.25 & 3.000 \\
Confidence Threshold      & 3.00 & 4.25 & 3.625 \\
Reliability Only          & 3.50 & 5.25 & 4.375 \\
Static Delegation         & 4.25 & 5.25 & 4.750 \\
\bottomrule
\end{tabular}%
}
\end{table}

The benchmark results support the central interpretation of the paper: the value of the Governance-Aware POMDP is not that it dominates every specialized heuristic in every stationary environment, but that it provides a robust adaptive policy for allocating AI decision authority when AI quality is uncertain, unstable, or evolving. This comparative evidence complements the structural validation results above and motivates the broader managerial discussion in Section~\ref{sec:discussion}.

\begin{table}[H]
\centering
\small
\resizebox{\textwidth}{!}{%
\begin{tabular}{p{0.24\textwidth}p{0.19\textwidth}p{0.15\textwidth}p{0.34\textwidth}}
\toprule
Property & Experiment & Status & Evidence \\
\midrule
Belief-state recovery & All synthetic scenarios & Supported & Mean posterior mode accuracy = 54.6\% \\
Reference-model shrinkage & Bad AI & Supported & Governed minus Raw AI CAGR = 1.74\% \\
Convergence toward validated AI & Perfect AI & Supported & Governed minus Raw AI CAGR = -2.42\% \\
Sequential validation / trust recovery & Improving AI & Supported & Governed minus Reference CAGR = -0.39\% \\
Downside protection under fragile AI & Fragile AI failure & Supported & Max drawdown improvement vs Raw AI = 18.65\% \\
Monotone governance in BaR & BaR sweep & Supported & Change in optimized AI exposure = -0.250 \\
Reported-confidence robustness & Confidence-only perturbation & Supported & Exposure range = 0.000000 \\
Forecast accuracy validation & Evidence-content sweep & Supported & Accuracy-exposure slope = 1.045 \\
Governance appetite calibration & BaR-penalty sensitivity & Supported & AI exposure decreases as $\lambda_{BaR}$ increases \\
Early warning & Fragile AI failure & Supported & Lead time in simulation steps = 398 steps \\
\bottomrule
\end{tabular}%
}
\caption{Theory validation summary. The table links each structural property to a validation experiment and records the diagnostic evidence used to assess whether the governance policy behaves consistently with the theory.}
\label{tab:theorem_validation}
\end{table}

\begin{table}[htbp]
\centering
\caption{Monte Carlo summary with bootstrap confidence intervals and drawdown diagnostics.}
\label{tab:monte_carlo_summary}
\resizebox{\textwidth}{!}{%
\begin{tabular}{lrrrrrrrrr}
\toprule
Scenario & Runs & \makecell{Reference\\CAGR\\Mean} & \makecell{Ungoverned AI\\CAGR\\Mean} & \makecell{Governed\\CAGR\\Mean} & \makecell{Governed\\CAGR\\95\\\% CI Low} & \makecell{Governed\\CAGR\\95\\\% CI High} & \makecell{Reference\\Max\\Drawdown} & \makecell{Ungoverned AI\\Max\\Drawdown} & \makecell{Governed\\Max\\Drawdown} \\
\midrule
Bad AI & 500.00 & 0.020 & -0.004 & 0.013 & 0.009 & 0.017 & -0.100 & -0.218 & -0.093 \\
Fragile AI Failure & 500.00 & 0.053 & -0.025 & 0.056 & 0.050 & 0.063 & -0.118 & -0.348 & -0.125 \\
Improving AI & 500.00 & 0.041 & 0.075 & 0.050 & 0.045 & 0.056 & -0.104 & -0.151 & -0.099 \\
Perfect AI & 500.00 & 0.047 & 0.089 & 0.061 & 0.055 & 0.067 & -0.104 & -0.145 & -0.101 \\
\bottomrule
\end{tabular}%
}
\end{table}

\section{Managerial Implications and Discussion}
\label{sec:discussion}

The analytical framework developed in this paper suggests that the principal managerial challenge created by artificial intelligence is not prediction but governance. As AI systems become increasingly capable of generating probabilistic assessments and candidate actions, organizations must continuously determine how much authority should be delegated to AI-generated recommendations while preserving independent organizational control. The governance policy developed here provides one quantitative mechanism for addressing this problem by combining Bayesian inference, sequential optimization, and empirical validation within a unified decision framework. Rather than treating governance as a procedural activity performed after model development, the framework embeds governance directly within organizational decision making, allowing delegation to evolve continuously as evidence quality, uncertainty, and institutional objectives change.

A first implication concerns the role of governance within organizations. Existing AI governance frameworks emphasize documentation, transparency, explainability, and human oversight. These mechanisms remain necessary, especially within regulated industries, but they provide limited guidance regarding how organizations should determine the appropriate degree of reliance on AI under changing conditions. The governance policy complements these frameworks by introducing a quantitative decision layer that transforms probabilistic evidence into delegation decisions. Governance therefore becomes an operational capability that can be optimized, validated, calibrated, and continuously improved rather than a static compliance process.

A second implication concerns the distinction between confidence and trust. Many AI systems generate numerical confidence estimates together with their recommendations, and these quantities are often interpreted as indicators of reliability. The empirical analysis demonstrates, however, that governance should respond primarily to externally validated Bayesian evidence rather than internally generated confidence declarations. The reported-confidence robustness experiments show that delegation remains largely unchanged when confidence alone varies, whereas forecast-accuracy validation demonstrates systematic changes in delegated authority as evidence quality improves. Organizations relying primarily on confidence thresholds may therefore allocate decision authority inefficiently when AI systems become systematically overconfident or underconfident.

A third implication concerns organizational adaptability. The governance policy continuously reallocates authority as new information becomes available, allowing organizations to respond dynamically to improvements or deterioration in AI capability. This interpretation is consistent with the literature on organizational learning and dynamic capabilities, which argues that sustained performance depends on the ability to reconfigure organizational resources as environments change \citep{teece1997dynamic,eisenhardt2000dynamic}. Within the present framework, delegated decision authority becomes one such resource. Bayesian learning updates the informational environment, Belief-at-Risk summarizes governance uncertainty, and the governance policy adjusts delegation accordingly.

The benchmark experiments further clarify when governance complexity is justified. Simple policies can be effective under stationary conditions: for example, Bayesian Shrinkage performs well when AI quality is persistently poor because rapid contraction toward the reference process is sufficient. However, the benchmark results show that the Governance-Aware POMDP provides the strongest general-purpose governance performance when AI quality is fragile, improving, or otherwise dynamic. This distinction is important for implementation. Organizations do not need complex governance machinery merely to reject consistently poor AI; the value of sequential Bayesian governance arises when AI-generated intelligence is potentially useful but uncertain, changing, and consequential.

The framework also has implications for organizational design. Classical theories of decision rights emphasize that effective organizations allocate authority according to information, incentives, and expertise \citep{simon1947,jensen1995,aghion1997}. The proposed governance policy extends these ideas by treating AI as an organizational actor whose authority is neither fixed nor binary. Delegated authority evolves according to posterior evidence quality and governance risk. This interpretation provides a bridge between organizational economics and contemporary AI governance by replacing static delegation rules with an adaptive Bayesian decision policy.

\paragraph{Benchmark implications.}
The comparative benchmark experiments provide an additional managerial perspective. Rather than evaluating the proposed Governance-Aware POMDP in isolation, the synthetic laboratory compares it with representative governance mechanisms including Static Delegation, Confidence Threshold, Reliability-Only Delegation, Bayesian Shrinkage, and SR11-7 style governance under identical Bayesian beliefs, information, simulated environments, and governance objectives. The results demonstrate that no single governance policy is universally optimal. Bayesian Shrinkage performs best when AI quality is persistently poor because rapid contraction toward the validated reference process minimizes unnecessary AI exposure. In contrast, the Governance-Aware POMDP achieves the highest governance utility across fragile, improving, and high-quality AI regimes, indicating that sequential Bayesian governance provides the strongest general-purpose policy when AI quality is uncertain, evolving, or unstable. These findings suggest that governance complexity should match environmental complexity: simple shrinkage rules are effective under stationary conditions, whereas adaptive organizational delegation becomes increasingly valuable as AI capability changes over time.

\paragraph{Organizational capability.}
Viewed from a Management Science perspective, the principal contribution is therefore not a new prediction algorithm but a quantitative theory of adaptive organizational delegation. Bayesian inference estimates the informational state, governance determines the appropriate allocation of decision authority, and execution remains under institutional control. This separation transforms governance from a procedural compliance activity into a measurable organizational capability that can be optimized, benchmarked, calibrated, and validated empirically.

Several limitations remain. The current implementation optimizes over a finite governance action space to preserve interpretability and computational tractability. Although this design facilitates practical deployment, future work may investigate continuous governance controls and richer organizational state representations. The forecast-accuracy validation experiments are controlled perturbations of evidence quality rather than direct measurements of ex ante predictive skill. Finally, while historical replay demonstrates consistent behaviour using observed market data and cached AI outputs, future work should evaluate the policy using continuously operating production AI systems and multiple interacting agents.

Taken together, the theoretical development and empirical evidence suggest that the principal managerial challenge created by AI is not whether AI should be used, but how organizations should allocate decision authority when AI-generated intelligence is valuable but uncertain. The governance policy proposed in this paper demonstrates that this problem can be formulated as a Bayesian sequential optimization problem whose solution adapts continuously to evidence quality, governance risk, and institutional objectives.

\section{Historical Market Replay}
\label{sec:historical_validation}

Historical replay complements the synthetic laboratory by evaluating the same governance logic using observed market data and cached AI outputs. The purpose of the replay is not to claim that a particular historical strategy is optimal, but to examine whether the governance policy exhibits the same structural behaviour outside a fully controlled simulation environment. The historical validation mirrors the synthetic methodology: representative stress tests, reported LLM-confidence robustness, forecast-accuracy validation, and governance-appetite sensitivity. This design distinguishes behaviour arising from the governance architecture from behaviour induced by any particular synthetic data-generating process and follows the validation roadmap in Table~\ref{tab:validation_roadmap}.

In the historical validation, cached AI outputs are perturbed to construct low-reliability, noisy, improving, and high-reliability operating conditions. Reported-confidence robustness is tested by varying only the confidence field while keeping evidence and recommendations fixed. Forecast-accuracy validation perturbs evidence content and recommendations while holding reported confidence fixed. Governance-appetite sensitivity varies the Belief-at-Risk penalty without changing the AI cache or market data. These experiments therefore evaluate whether the policy behaves consistently when confronted with historical market structure rather than synthetic regimes alone.

\subsection{Historical Stress Tests}
Table~\ref{tab:historical_validation_summary} summarizes the historical stress-test results. The governance policy assigns substantially lower delegated authority under low-reliability evidence and higher authority under high-reliability evidence, while Belief-at-Risk remains negatively associated with delegation. This pattern is consistent with the synthetic results and supports the interpretation that governance reallocates authority according to evidence quality and governance risk.

\begin{table}[H]
\centering
\small
\resizebox{\textwidth}{!}{%
\begin{tabular}{lrrrrrr}
\toprule
Scenario & Exposure & Reject rate & Governed Sharpe & Raw AI Sharpe & $\Delta$ Sharpe vs AI & Mean BaR \\
\midrule
Low Reliability & 0.174 & 0.645 & 0.993 & 0.832 & 0.161 & 0.0181 \\
Noisy & 0.405 & 0.017 & 0.527 & 1.217 & -0.690 & 0.0098 \\
Improving & 0.368 & 0.176 & 1.028 & 0.986 & 0.042 & 0.0101 \\
High Reliability & 0.539 & 0.062 & 0.393 & 1.556 & -1.162 & 0.0093 \\
\bottomrule
\end{tabular}%
}
\caption{Historical market replay stress-test summary. Governance reduces exposure and increases rejection under degraded AI conditions while preserving a reference-model comparison.}
\label{tab:historical_validation_summary}
\end{table}

\begin{figure}[H]
\centering
\IfFileExists{figures/historical_adaptive_trust_allocation_five_panel.png}{\includegraphics[width=0.95\textwidth]{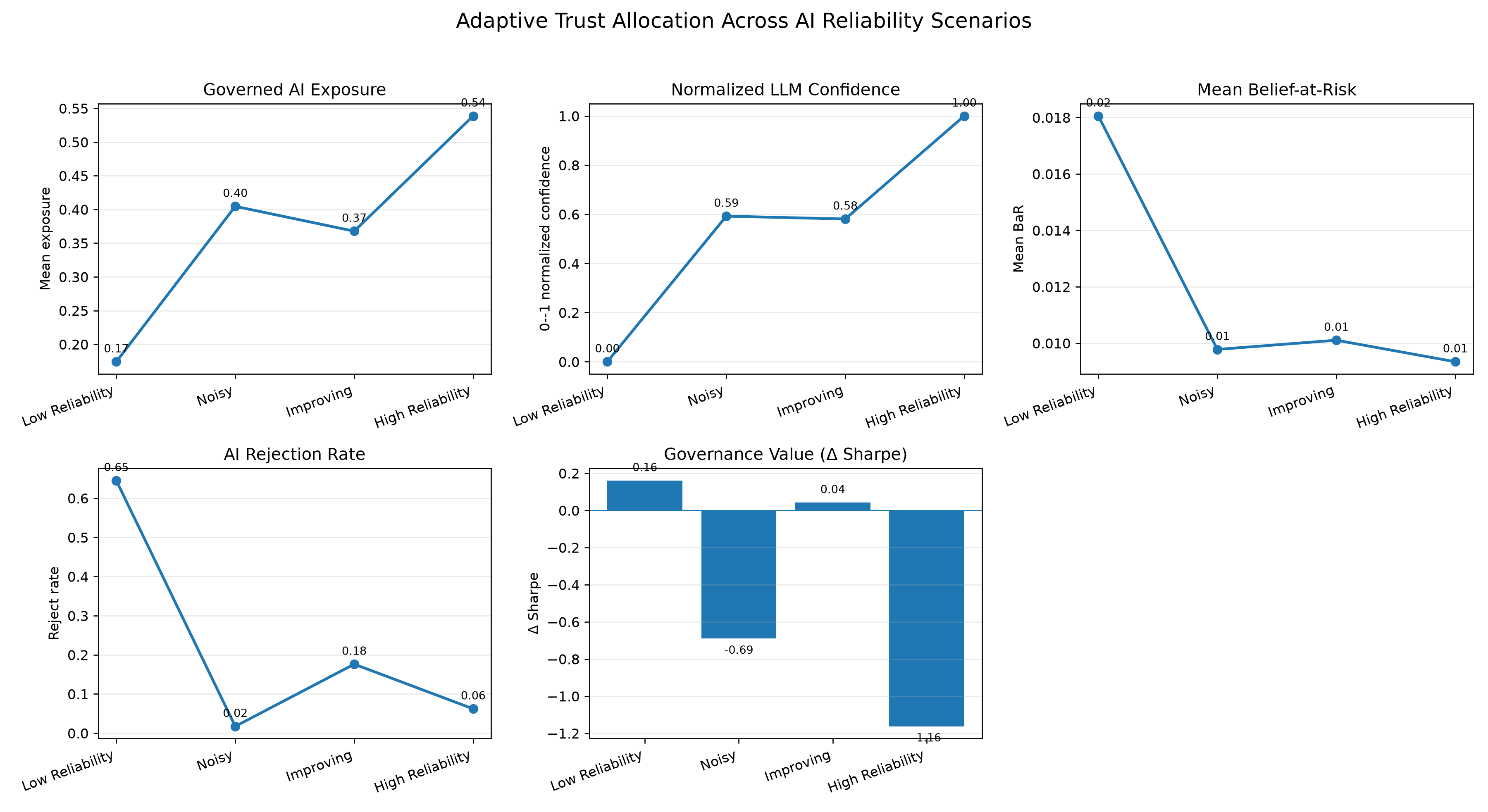}}{\fbox{Missing historical adaptive trust figure.}}
\caption{Historical adaptive delegation across AI operating conditions. Controlled cached-AI perturbations show that delegated authority, rejection, Belief-at-Risk, and governance value change coherently across degraded, noisy, improving, and high-reliability operating regimes. The figure provides historical evidence that the governance policy responds to evidence quality rather than merely reproducing the behaviour of the raw AI signal.}
\label{fig:historical_adaptive_trust}
\end{figure}

\subsection{Historical Reported LLM-Confidence Robustness}
Table~\ref{tab:historical_confidence_robustness} and Figure~\ref{fig:historical_confidence_robustness} show that varying only reported confidence leaves delegated authority, Belief-at-Risk, and rejection essentially unchanged. This result is important because it demonstrates that the separation between confidence and trust is not merely an artifact of the synthetic laboratory.

\begin{table}[H]
\centering
\small
\resizebox{0.85\textwidth}{!}{%
\begin{tabular}{lrrr}
\toprule
Diagnostic & Start & End & Range \\
\midrule
Reported LLM confidence & 0.031 & 0.232 & 0.201 \\
Governed AI exposure & 0.761 & 0.761 & 0.000000 \\
Belief-at-Risk & 0.0102 & 0.0102 & 0.000000 \\
Reject rate & 0.112 & 0.112 & 0.000000 \\
\bottomrule
\end{tabular}%
}
\caption{Historical reported LLM-confidence robustness summary. Varying only reported confidence leaves governance exposure, Belief-at-Risk, and rejection essentially unchanged, indicating that the controller is not mechanically driven by self-reported confidence.}
\label{tab:historical_confidence_robustness}
\end{table}

\begin{figure}[H]
\centering
\IfFileExists{figures/historical_reported_confidence_robustness.png}{\includegraphics[width=0.95\textwidth]{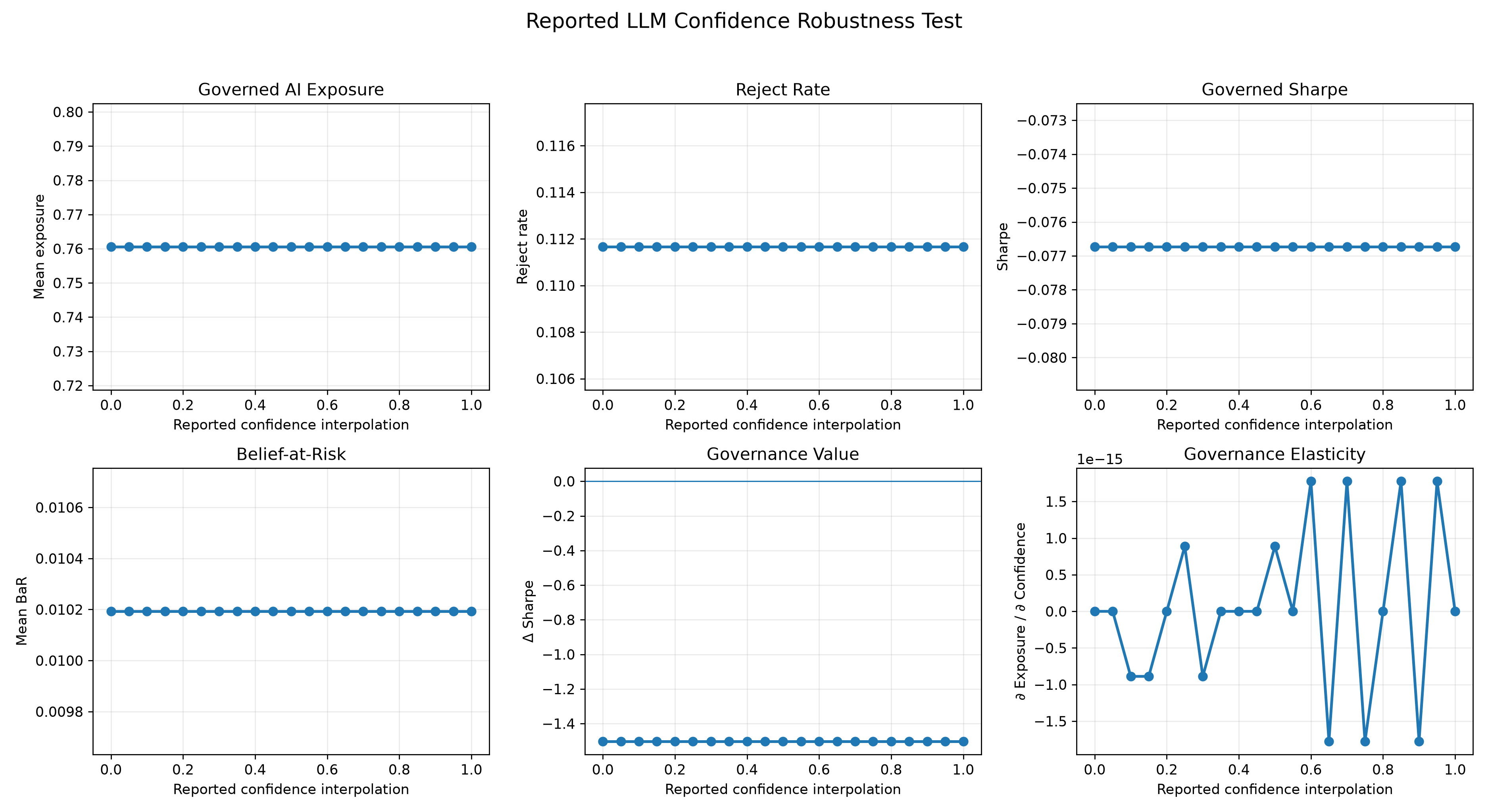}}{\fbox{Missing historical confidence robustness figure.}}
\caption{Historical reported LLM-confidence robustness. Confidence-only perturbations leave the governance policy nearly invariant when evidence and recommendations are fixed, supporting the interpretation that organizational trust is not mechanically determined by AI self-confidence.}
\label{fig:historical_confidence_robustness}
\end{figure}

\subsection{Historical Forecast Accuracy Validation}
Table~\ref{tab:historical_forecast_accuracy} and Figure~\ref{fig:historical_forecast_accuracy} report the primary historical adaptive-delegation result. Reported confidence is held fixed while evidence content and recommendations are perturbed from degraded to original. Delegated authority rises substantially and Belief-at-Risk declines as evidence quality improves.
\begin{table}[H]
\centering
\small
\caption{Historical forecast-accuracy validation summary. Evidence quality changes while reported confidence is held fixed; delegated authority rises and Belief-at-Risk declines, supporting adaptive delegation.}
\label{tab:historical_forecast_accuracy}
\resizebox{\textwidth}{!}{%
\begin{tabular}{lrr}
\toprule
Diagnostic & Low evidence quality & High evidence quality \\
\midrule
Forecast-accuracy interpolation & 0.00 & 1.00 \\
Fixed reported confidence & 0.626 & 0.626 \\
Delegated AI authority & 0.262 & 0.761 \\
Reject rate & 0.268 & 0.112 \\
Belief-at-Risk & 0.0147 & 0.0102 \\
Governed Sharpe & 1.540 & -0.077 \\
\midrule
Delegation slope & \multicolumn{2}{c}{0.5086} \\
$R^2$ & \multicolumn{2}{c}{0.9631} \\
Spearman rank correlation & \multicolumn{2}{c}{1.0000} \\
\bottomrule
\end{tabular}%
}
\end{table}

\begin{figure}[H]
\centering
\IfFileExists{figures/historical_forecast_accuracy_validation.png}{\includegraphics[width=0.95\textwidth]{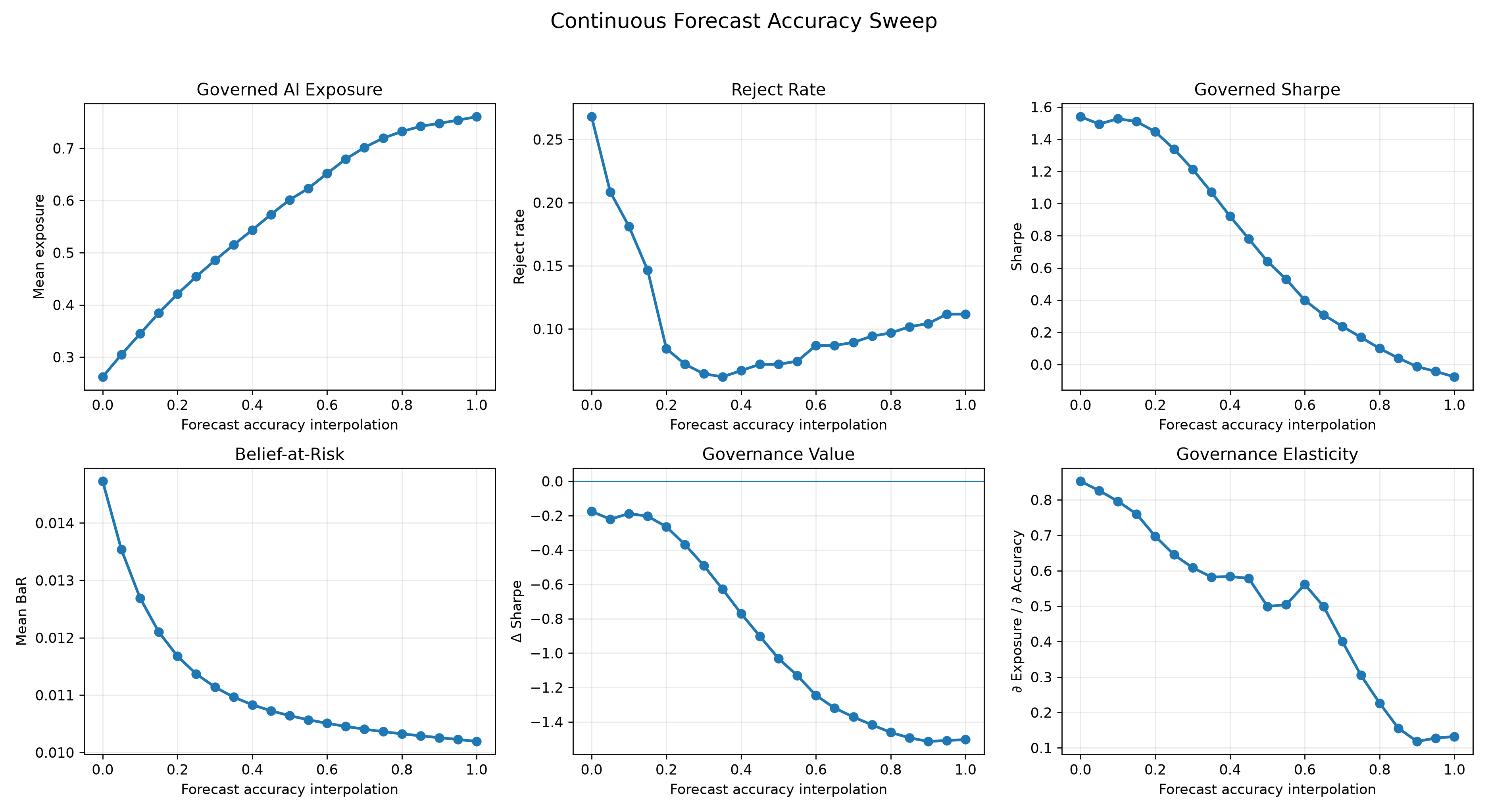}}{\fbox{Missing historical forecast accuracy figure.}}
\caption{Historical forecast-accuracy validation. Delegated authority increases with evidence quality while reported confidence is held approximately fixed. The result supports adaptive delegation based on Bayesian evidence rather than confidence declarations.}
\label{fig:historical_forecast_accuracy}
\end{figure}

\begin{figure}[H]
\centering
\IfFileExists{figures/historical_reported_confidence_vs_accuracy_validation.png}{\includegraphics[width=0.95\textwidth]{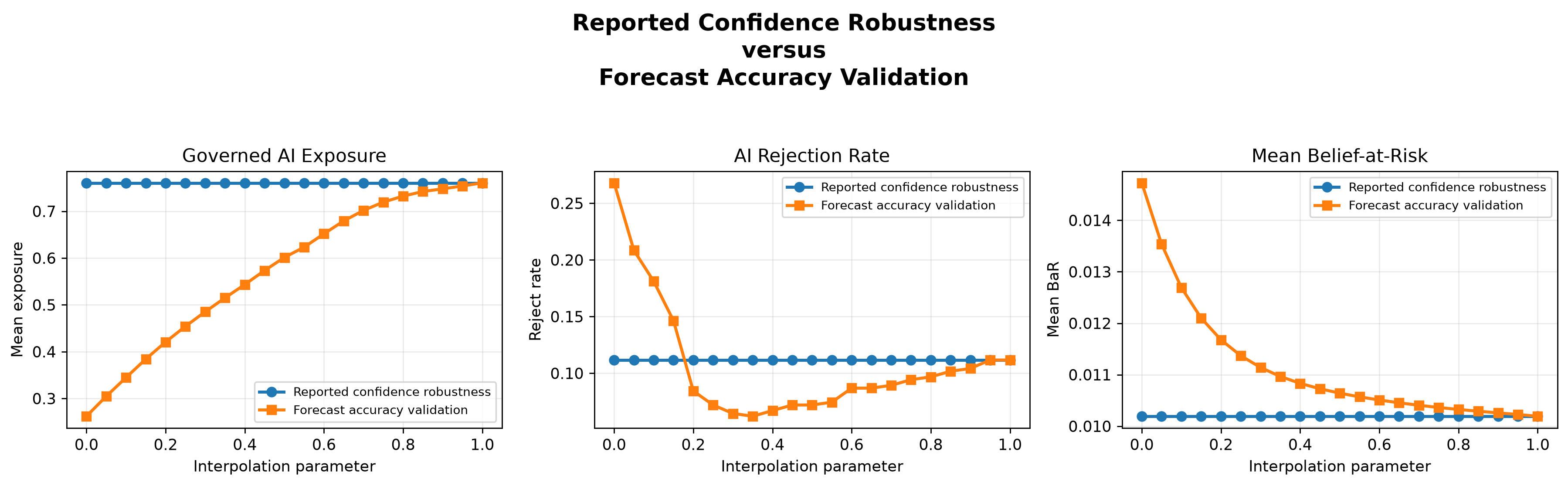}}{\fbox{Missing historical comparison figure.}}
\caption{Historical comparison of reported confidence robustness and forecast-accuracy validation. Confidence-only perturbations generate negligible changes in governance behaviour, whereas evidence-quality perturbations generate systematic changes in delegated authority, rejection, and Belief-at-Risk.}
\label{fig:historical_confidence_vs_accuracy}
\end{figure}

\subsection{Historical Governance Appetite}
Table~\ref{tab:historical_governance_appetite} and Figure~\ref{fig:historical_governance_appetite} show that increasing the BaR penalty produces a more conservative governance policy. This result supports the claim that governance appetite can be calibrated through transparent policy parameters rather than through redesign of the decision process.

\begin{table}[H]
\centering
\small
\resizebox{\textwidth}{!}{%
\begin{tabular}{rrrrrr}
\toprule
$\lambda_{BaR}$ & Exposure & Reject rate & Approve rate & Governed Sharpe & $\Delta$ Sharpe vs AI \\
\midrule
0.5 & 0.746 & 0.020 & 0.342 & 0.625 & -0.930 \\
1.0 & 0.675 & 0.030 & 0.253 & 0.498 & -1.057 \\
1.5 & 0.620 & 0.042 & 0.203 & 0.441 & -1.115 \\
2.5 & 0.539 & 0.062 & 0.156 & 0.393 & -1.162 \\
\bottomrule
\end{tabular}%
}
\caption{Historical governance-appetite sensitivity. Increasing the Belief-at-Risk penalty produces a more conservative governance policy without changing the model structure.}
\label{tab:historical_governance_appetite}
\end{table}

\begin{figure}[H]
\centering
\IfFileExists{figures/historical_governance_appetite_sensitivity.png}{\includegraphics[width=0.95\textwidth]{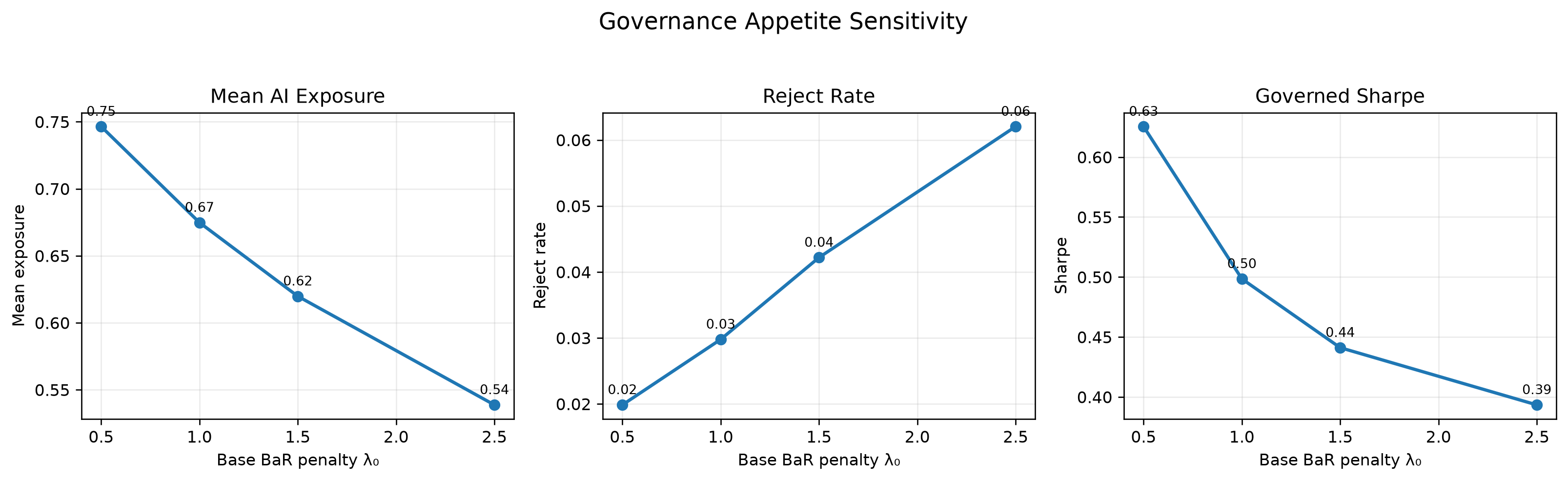}}{\fbox{Missing historical governance appetite figure.}}
\caption{Historical governance-appetite sensitivity. Increasing the Belief-at-Risk penalty reduces delegated authority and shifts the policy toward more conservative governance without changing the AI outputs, reference process, or market data.}
\label{fig:historical_governance_appetite}
\end{figure}

\section{Conclusion}
\label{sec:conclusion}

\textbf{Summary.} Artificial intelligence changes not only how organizations process information but also how they allocate decision authority. As AI systems become increasingly capable of generating probabilistic assessments and candidate actions, organizations face a new managerial problem: determining how much authority should be delegated to AI-generated intelligence when evidence quality, uncertainty, and governance risk evolve through time. This paper has argued that this problem is fundamentally one of adaptive organizational delegation rather than prediction alone.

\textbf{Framework.} To address this challenge, we developed a Governance-Aware Partially Observable Markov Decision Process that separates inference, validation, governance, and execution. Bayesian filtering transforms heterogeneous observations into posterior beliefs over latent organizational states, Belief-at-Risk summarizes governance uncertainty through posterior uncertainty, belief instability, and downside consequence, and an approximate Bellman recursion determines the degree of delegated authority assigned to AI-generated recommendations. Unlike conventional applications of POMDPs, the optimization variable is delegated organizational authority rather than the operational action itself.

\textbf{Theory.} The paper established several structural properties of the proposed governance framework, including existence of an optimal governance policy, Bayesian belief-state sufficiency, monotone adaptation to increasing governance risk, robustness to confidence-only perturbations, adaptive delegation under improving evidence quality, and preservation of a validated reference-process fallback. Together, these results provide a decision-theoretic foundation for adaptive AI governance under uncertainty.

\textbf{Validation.} The empirical analysis complements the theoretical development through both synthetic governance experiments and historical market replay. The synthetic laboratory demonstrates that the governance policy exhibits graceful degradation under deteriorating AI quality, robustness to confidence-only perturbations, adaptive delegation as Bayesian evidence improves, interpretable governance-appetite calibration, and early-warning behaviour under fragile AI failure. Historical replay shows that these qualitative governance properties persist when the same governance architecture is evaluated using observed market data.

\textbf{Benchmarking.} A distinguishing contribution of this paper is the introduction of a comparative benchmarking framework for AI governance policies. Five representative governance mechanisms?Static Delegation, Confidence Threshold, Reliability-Only Delegation, Bayesian Shrinkage, and SR11-7 Style Governance?were evaluated under identical Bayesian beliefs, information, simulated environments, and governance objectives. The benchmark demonstrates that no single governance policy is universally optimal. Bayesian Shrinkage performs best when AI quality is persistently poor because rapid contraction toward the validated reference process minimizes unnecessary AI exposure. In contrast, the Governance-Aware POMDP achieves the highest governance utility across fragile, improving, and high-quality AI regimes, demonstrating that sequential Bayesian governance provides the strongest general-purpose governance policy when AI quality is uncertain, evolving, or unstable.

\textbf{Managerial Implications.} These findings suggest that governance should not be interpreted as maximizing AI utilization. Rather, governance should determine how much organizational authority should be delegated to AI-generated intelligence under prevailing uncertainty. In stationary environments, relatively simple governance rules may be sufficient. However, as AI systems become adaptive, continuously updated, and increasingly integrated into organizational workflows, sequential governance becomes progressively more valuable because delegated authority must evolve together with evidence quality and organizational objectives.

\textbf{Broader Perspective.} More broadly, this paper suggests that adaptive AI delegation may become as important to AI-enabled organizations as quantitative risk management has become to modern financial institutions. Rather than relying on static approval rules or heuristic confidence thresholds, organizations can treat delegated authority itself as an object of optimization. Bayesian inference estimates the evolving informational state, governance determines the economically appropriate degree of delegation, and organizational execution remains under institutional control. This perspective transforms AI governance from a largely procedural activity into a quantitative organizational capability that can be optimized, benchmarked, calibrated, and empirically validated.

\textbf{Future Research.} Several directions remain open. Future work may investigate continuous governance action spaces, richer latent-state representations, multi-agent governance problems, fully Bayesian parameter learning, and governance under strategic interaction among multiple AI systems. Extending the benchmarking framework to additional governance mechanisms and validating adaptive delegation across multiple industries represent particularly promising directions. More generally, the proposed framework provides a quantitative foundation upon which future research can build increasingly sophisticated theories of adaptive organizational delegation under uncertainty.

\appendix

\section{Supplementary Proof Details}
\label{app:proofs}

This appendix provides additional detail for the structural results in Section~\ref{sec:theory}. The main text states the results in organizational language; the appendix records the probability and optimization arguments explicitly.

\subsection{Filtering and Sufficiency}
Let $\mathcal F_t$ denote the sigma-field generated by the observation history, AI evidence, realized outcomes used for validation, and recursively updated governance diagnostics up to time $t$. Under the Markov transition model, the posterior belief $b_t=P(X_t\mid\mathcal F_t)$ satisfies the prediction-update recursion used in the main text. For every bounded measurable function $h$, the conditional expectation of the next latent state satisfies
\[
\E[h(X_{t+1})\mid\mathcal F_t]=\int h(x')\left\{\int P(dx'\mid x)b_t(dx)\right\}.
\]
In the finite-state implementation this reduces to $\sum_i b_t(i)\sum_j P_{ij}h(j)$. Hence the complete history affects future latent-state predictions only through $b_t$. If the additional governance diagnostics are themselves updated recursively, the augmented governance state is Markov. This is the precise sense in which posterior beliefs are sufficient for governance optimization.

\subsection{Bellman Operator}
The Bellman operator used in the paper acts on bounded measurable value functions. Boundedness of $r_G$ ensures that $\mathcal T$ maps bounded functions into bounded functions. The contraction proof in Theorem~\ref{thm:existence} relies only on the discount factor and the fact that conditional expectation is a nonexpansive operator under the supremum norm. Thus the existence result is not a numerical artifact of the implementation; it is a standard fixed-point property of discounted dynamic programming.

\subsection{Monotone Comparative Statics}
The monotonicity results use finite-action comparative statics. If higher governance actions correspond to higher delegation weights and the Bellman objective has decreasing differences in governance risk and action, then increases in governance risk cannot make higher-delegation actions relatively more attractive. The direct BaR penalty produces exactly this decreasing-differences structure because the marginal cost of higher delegation increases with $\BaR$. The theorem is stated as a weak monotonicity result for the optimal correspondence because finite action sets may create ties. A conservative implementation can break ties toward lower delegation without violating optimality.

\subsection{Delegation as Adaptive Shrinkage}
The executed decision may be written as
\[
u_t^{Exec}=u_t^{Ref}+\alpha_t(u_t^{AI}-u_t^{Ref}).
\]
Thus the governance policy is an adaptive shrinkage rule toward the independently validated reference process. Classical shrinkage uses a fixed or estimated shrinkage intensity; here the shrinkage factor is chosen by a governance policy whose state includes posterior beliefs, governance risk, reliability diagnostics, and recommendations. This representation explains why admissibility follows from convexity and why complete rejection of AI corresponds exactly to the reference process.

\subsection{Reported Confidence}
Reported confidence can influence governance if it changes Bayesian beliefs, recommendations, reliability diagnostics, consequences, or transition assessments. The confidence-robustness theorem therefore should not be interpreted as claiming that confidence is irrelevant in all settings. It establishes the narrower property tested in the validation experiment: if confidence is perturbed while the validated evidence and recommendation content are fixed, organizational delegation should remain unchanged. This property is the formal counterpart of the managerial claim that confidence is not the same object as trust.

\section{Additional Experimental Results}
\label{app:additional_experiments}

This appendix reports supplementary outputs from the synthetic validation laboratory. To avoid duplicating the main text, scenario growth figures already discussed in Section~\ref{sec:synthetic_validation} are not repeated here. The appendix instead reports additional Monte Carlo diagnostics and provides file-level notes for reproducibility.

\subsection{Supplementary Tables}

\IfFileExists{tables/summary_metrics.tex}{

}{
\begin{table}[htbp]
\centering
\caption{Scenario-level summary metrics.}
\label{tab:app_summary_metrics}
\fbox{\parbox{0.85\textwidth}{Missing table: \texttt{tables/summary\_metrics.tex}}}
\end{table}
}

\subsection{Monte Carlo Diagnostics}

The Monte Carlo diagnostics provide a robustness check on the synthetic validation results. Whereas the main text reports the principal Monte Carlo table, the figures below show how governed performance varies across independent simulation paths and bootstrap resamples.

\begin{figure}[htbp]
\centering
\IfFileExists{figures/mc_mean_cagr_by_scenario.png}{
\includegraphics[width=0.82\textwidth]{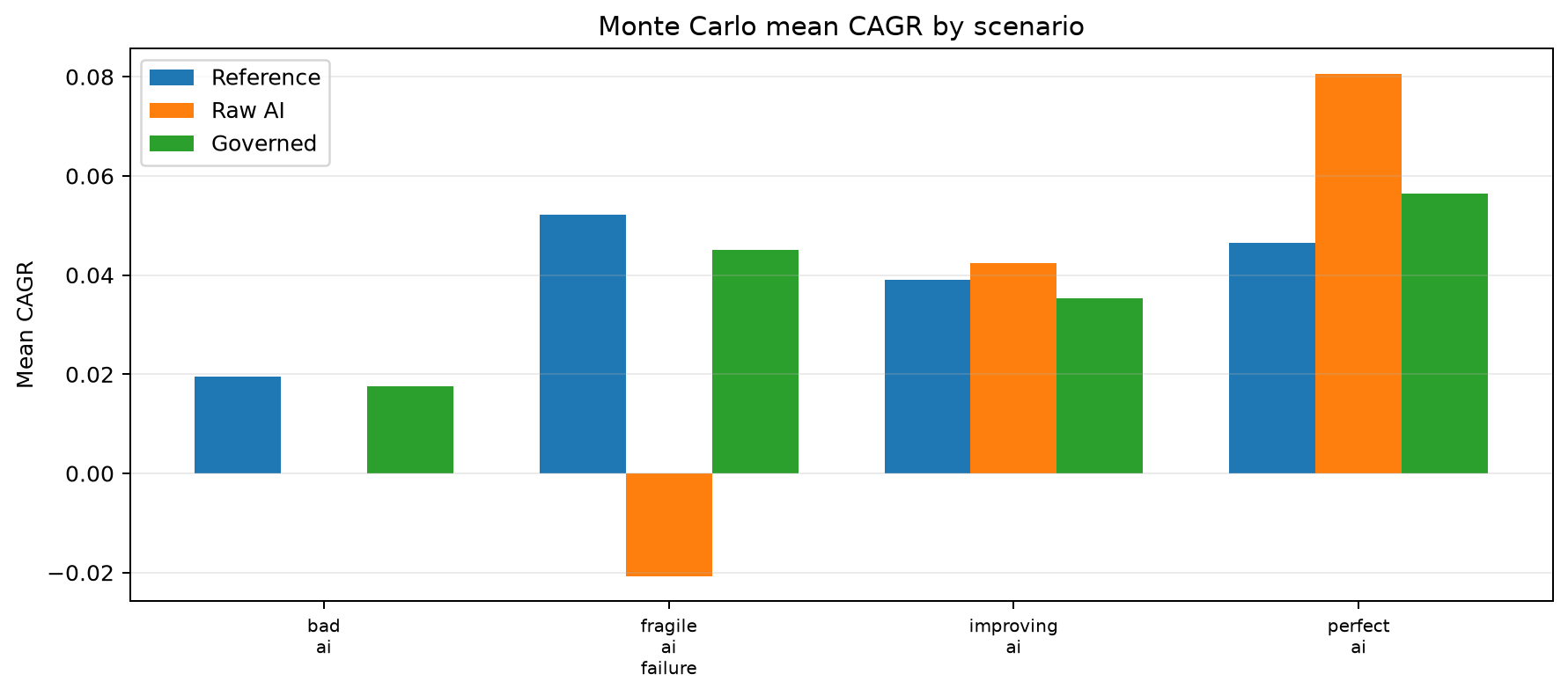}
}{
\fbox{\parbox{0.80\textwidth}{Missing figure: \texttt{figures/mc\_mean\_cagr\_by\_scenario.png}}}
}
\caption{Monte Carlo mean CAGR by scenario. The figure summarizes average performance across independent synthetic realizations.}
\label{fig:app_mc_mean_cagr}
\end{figure}

\begin{figure}[htbp]
\centering
\IfFileExists{figures/mc_convergence_governed_cagr.png}{
\includegraphics[width=0.82\textwidth]{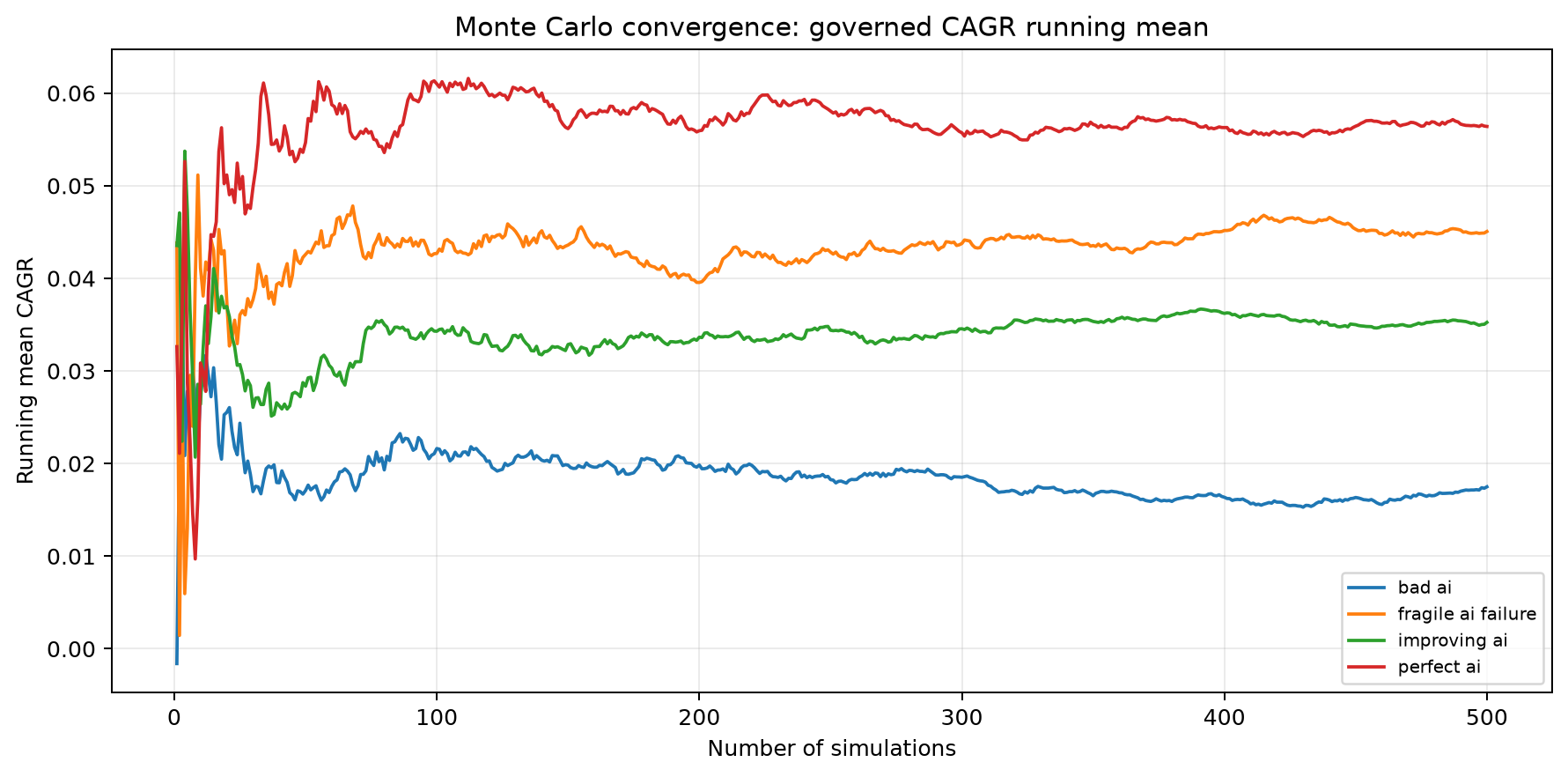}
}{
\fbox{\parbox{0.80\textwidth}{Missing figure: \texttt{figures/mc\_convergence\_governed\_cagr.png}}}
}
\caption{Monte Carlo convergence diagnostic for governed CAGR. The curve verifies that the reported governed performance is not driven by a small number of simulation paths.}
\label{fig:app_mc_convergence}
\end{figure}

\begin{figure}[htbp]
\centering
\IfFileExists{figures/bootstrap_ci_governed_cagr.png}{
\includegraphics[width=0.82\textwidth]{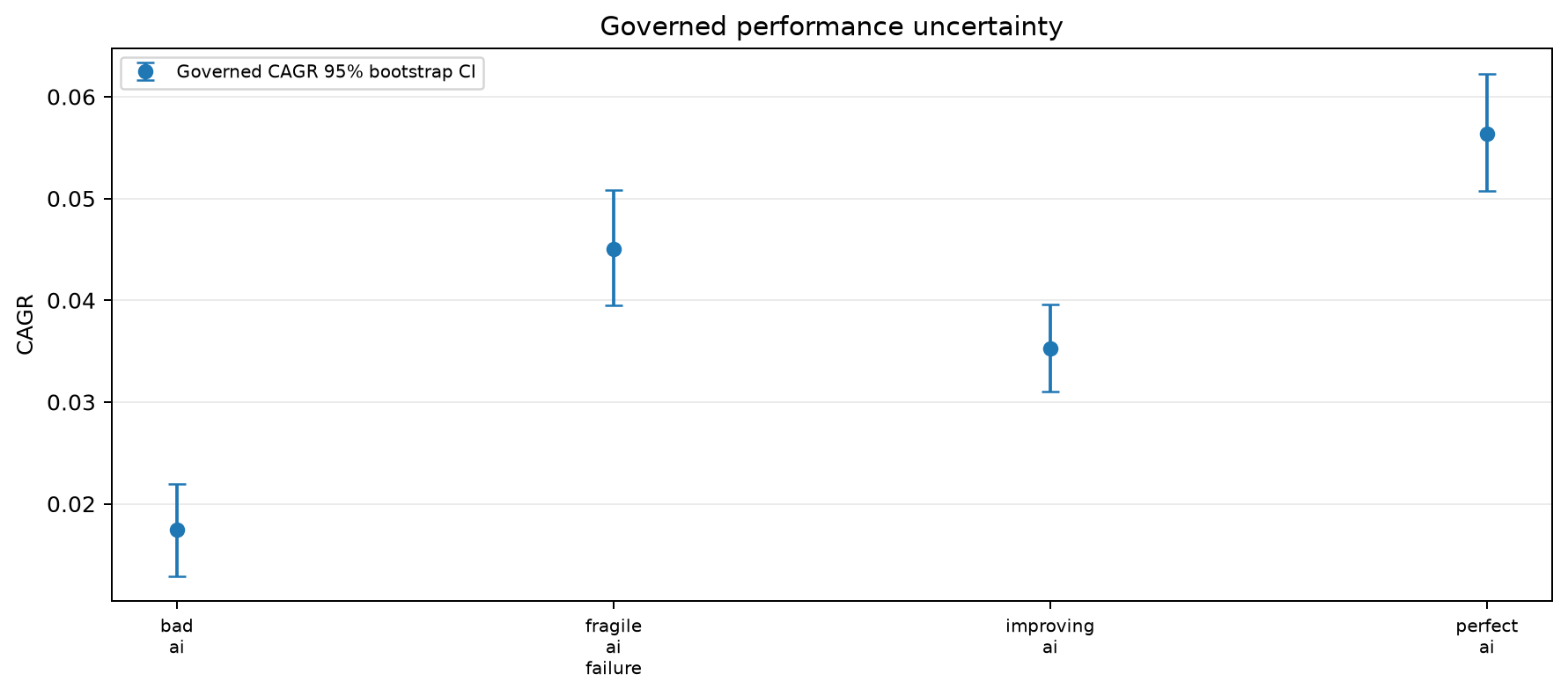}
}{
\fbox{\parbox{0.80\textwidth}{Missing figure: \texttt{figures/bootstrap\_ci\_governed\_cagr.png}}}
}
\caption{Bootstrap confidence intervals for governed CAGR. The intervals summarize uncertainty in governed performance across bootstrap replications.}
\label{fig:app_bootstrap_ci}
\end{figure}

\subsection{Reproducibility Notes}

All figures in the main text and appendix are generated automatically by the synthetic validation laboratory. The final manuscript intentionally excludes earlier placeholder figures such as \texttt{ablation\_tau\_modes.png}, \texttt{sensitivity\_ewma.png}, \texttt{fragile\_ai\_failure\_02\_diagnostics.png}, and \texttt{fragile\_ai\_failure\_03\_exposure.png}, because these were not part of the final figure set and are not required for the theorem-validation narrative.

\section{Computational Details}

\paragraph{Algorithm 1. Governance-Aware Bayesian Decision Process.}
At each decision time $t$:
\begin{enumerate}[leftmargin=*]
\item Receive observations $O_t$.
\item Compute AI outputs $(q_t,u_t^{AI},c_t)=\mathcal L_\theta(O_t)$.
\item Predict the prior belief $\Pi_t=P^\top b_{t-1}$.
\item Update the posterior belief $b_t(i)=q_t(i)\Pi_t(i)/\sum_j q_t(j)\Pi_t(j)$.
\item Compute entropy, belief drift, posterior adverse probability, recommendation disagreement, and $\BaR_t$.
\item Construct the governance state $s_t^G=(b_t,u_t^{AI},u_t^{Ref},\BaR_t,c_t,\eta_t)$.
\item Evaluate the finite governance action set $\G$ using the governance utility $r_G(s_t^G,g_t)$ and approximate Bellman value.
\item Select $g_t=\pi_G^\star(s_t^G)$ and compute $\alpha_t=\alpha(g_t)$.
\item Execute $u_t^{Exec}=\alpha_tu_t^{AI}+(1-\alpha_t)u_t^{Ref}$.
\item Store $(b_t,g_t,\BaR_t,u_t^{Exec})$ for the next decision epoch.
\end{enumerate}

The computational complexity of each governance step is dominated by Bayesian filtering over $K$ latent states and evaluation of a small finite governance action set. This makes the method suitable for replay validation and runtime monitoring.

\bibliographystyle{apalike}
\bibliography{main}

\end{document}